%% file: main.tex
\documentclass[11pt]{article}

\usepackage[letterpaper,margin=1.0in]{geometry}

\usepackage{cite}
\input{macros}
\usepackage{appendix}
\usepackage{ifthen}

\newcommand{\FullOrShort}{full}    % change “short” ↔ “full” here

\ifthenelse{\equal{\FullOrShort}{full}}{%
  \newcommand{\fullOnly}[1]{#1}
  \newcommand{\shortOnly}[1]{}     % drop short‐only bits

}{%
  \newcommand{\fullOnly}[1]{}      % drop full‐only bits
  \newcommand{\shortOnly}[1]{#1}   % keep short‐only bits

}

\begin{document}
\date{}
\title{Density-Dependent Graph Orientation and Coloring \\
in Scalable MPC}

\author{
   Mohsen Ghaffari \\
   \small{MIT}\\
   \small{ghaffari@mit.edu}
  \and
   Christoph Grunau \\
   \small{ETH Zurich}\\
   \small{cgrunau@inf.ethz.ch}
 }
\maketitle
 \begin{abstract}
This paper presents massively parallel computation (MPC) algorithms in the strongly sublinear memory regime (aka, scalable MPC) for orienting and coloring graphs as a function of its subgraph density. Our algorithms run in $\poly(\log\log n)$ rounds and compute an orientation of the edges with maximum outdegree $O(\alpha \log\log n)$ as well as a coloring of the vertices with $O(\alpha \log\log n)$ colors. Here, $\alpha$ denotes the density of the densest subgraph. Our algorithm's round complexity is notable because it breaks the $\tilde{\Theta}(\sqrt{\log n})$ barrier, which applied to the previously best known density-dependent orientation algorithm [Ghaffari, Lattanzi, and Mitrovic ICML'19] and is common to many other scalable MPC algorithms. 
\end{abstract}

\newpage
\setcounter{page}{0}
\thispagestyle{empty}
{\small \tableofcontents}

\newpage

\input{body}
\bibliographystyle{alpha}
\bibliography{ref}
\end{document}

%% file: macros.tex
\usepackage{amsmath, amsthm, amsfonts}
\usepackage{thmtools,thm-restate}
\usepackage{bbm, bm}
\usepackage{lipsum}

\usepackage{graphicx}
\usepackage{xcolor}

\usepackage{dsfont}

\usepackage[colorinlistoftodos, textsize=tiny]{todonotes}

\usepackage{framed}
\usepackage[framemethod=tikz]{mdframed}
\usepackage[bottom]{footmisc}
\usepackage{enumitem}
\setitemize{noitemsep,topsep=3pt,parsep=3pt,partopsep=3pt}
\usepackage[font=small]{caption}
\usepackage{xspace}
\usepackage[ruled,vlined]{algorithm2e}

\definecolor{darkgreen}{rgb}{0,0.5,0}
\definecolor{darkblue}{rgb}{0,0,0.6}
\usepackage{hyperref}
\hypersetup{
    unicode=false,          % non-Latin characters in Acrobat’s bookmarks
    colorlinks=true,        % false: boxed links; true: colored links
    linkcolor=darkblue,          % color of internal links (change box color with linkbordercolor)
    citecolor=darkgreen,        % color of links to bibliography
    filecolor=magenta,      % color of file links
    urlcolor=cyan           % color of external links
}
\usepackage[capitalize, nameinlink]{cleveref}
\Crefname{theorem}{Theorem}{Theorems}
\Crefname{lemma}{Lemma}{Lemmas}
\Crefname{claim}{Claim}{Claims}
\Crefname{remark}{Remark}{Remarks}
\Crefname{observation}{Observation}{Observations}

\newtheorem{theorem}{Theorem}[section]
\newtheorem{lemma}[theorem]{Lemma}
\newtheorem{meta-theorem}[theorem]{Meta-Theorem}
\newtheorem{claim}[theorem]{Claim}

\newtheorem{definition}{Definition}[section]

\newcommand{\eps}{\varepsilon}

\newcommand{\poly}{\operatorname{\text{{\rm poly}}}}

\renewcommand{\tilde}{\widetilde}

%% file: body.tex
\section{Introduction}
We present the first $\poly(\log\log n)$-round Scalable Massively Parallel Computation (MPC) algorithms for low out-degree orientation and coloring as a function of \textit{subgraph density} (or equivalently, arboricity). We first review the context and state of the art, and then state our results.

\subsection{Models and Basic Definitions}
\paragraph{Massively Parallel Computation Model.} We work with the massively parallel computation (MPC) model\cite{Karloff2010AMO, Goodrich2011SortingSA, beame2017communication, andoni2014parallel}, which has become de facto the standard theoretical abstraction for large-scale distributed data processing frameworks such as MapReduce\cite{Dean2004MapReduceSD}, Hadoop\cite{White2009HadoopTD}, Spark\cite{Zaharia2010SparkCC} and Dryad\cite{Isard2007DryadDD}. We provide a brief definition next; please see the recent book of Im et. al.~\cite{im2023massively} or the lecture notes of Ghaffari~\cite{GhaffariMPCNotes} for more.

The distributed system is composed of $M$ machines, and the data is distributed among them (arbitrarily). In the case of graph problems, this data is the input graph $G=(V, E)$, for which we also usually use the notations $m=|E|$ and $n=|V|$. Each machine is assumed to have a \textit{local memory} capacity of $S$ words. A word is simply $O(\log n)$ bits and can describe, for instance, a single vertex or edge. We usually refer to the total summation of memories, which is simply $MS$, as the \textit{global memory} of the system. Trivially, we need that $MS=\Omega(m+n)$.  Computation proceeds in synchronous rounds. Per round, each machine can perform some computation on the data it holds and then send some information to other machines. The only communication constraint is that the total amount of data that one machine sends or receives, in a round, cannot exceed its memory capacity $S$. The primary measure of interest is the round complexity of the algorithm. 

The setting where $S$ can be polynomially smaller than $n$ --- concretely $S \leq n^{\delta}$ for a constant $\delta\in(0,1)$ --- has been called the \textit{strongly sublinear memory regime}, and this indicates the most challenging (and also the most desirable) domain for algorithm design. Algorithms in this regime are frequently referred to as \textit{Scalable MPC algorithms}, and this is the focus of the present paper.

\paragraph{LOCAL Model.} Many state-of-the-art (scalable) MPC algorithms for graph problems are designed based on approaches in the more classic LOCAL model of distributed computing\cite{linial1987LOCAL,peleg00}. So we briefly recall this model: the graph $G=(V, E)$ represents an abstraction of the computer network, with one computer per node, and it has an $O(\log n)$-bit identifier. Per round, each computer/node can send one message to each of its neighbors in $G$. At the start, each node knows only its neighbors. At the end of the computation, each node should know its own part of the output, e.g., its color in a coloring. Again the primary measure is the round complexity.

\paragraph{Densest Subgraph Density, and Abroricity.} The \textit{maximum subgraph density} (often, simply \textit{subgraph density}) of a graph $G$ is defined as $\alpha(G) = \max_{S\subseteq V} \{\frac{|E(S)|}{|S|}\}$, where $E(S)$ denotes the set of edges with both endpoints in $S$. A closely related measure is the graph's arboricity $\lambda(G)$, defined as $\lambda(G)=\max_{S\subseteq V} \{\lceil\frac{|E(S)|}{|S|-1}\rceil\}$, which is also equivalent to the the minimum number of forests into which the edges $E$ can be partitioned~\cite{nash1961edge, nash1964decomposition, tutte1961problem}. Note that $\alpha(G) \leq \lambda(G) \leq \alpha(G)+1$. 

\paragraph{Low Out-degree Orientation.} We are interested in edge orientations with small maximum outdegree. In any orientation, the maximum outdegree is lower bounded by the subgraph density $\alpha(G)$. We want our outdegree bound to be close to this. Since $\alpha(G)$ and $\lambda(G)$ are equal up to a $+1$, our results can be stated in terms of either of these, and we use arboricity $\lambda=\lambda(G)$ from now on (also since the special case of $\lambda=1$ nicely coincides with forests).

% A more structured format for the orientation, which is achieved in our results and the prior work, and opens the road for applications in coloring, is as follows: an $L$-layer H-partition with outdegree $d$ is a defined as a partition of the set $V$ of vertices into $L$ layers $H_1, H_2, \dots, H_L$, such that each node $v\in H_i$ has at most $d$ neighbors in $H_{i} \cup \dots \cup H_{L}$. This implies an orientation by orienting edges toward the higher layer (breaking ties arbitrarily, e.g., according to identifiers).

% We will use low-outdegree orientation as a key tool in devising a coloring algorithm with a number of colors similar to the outdegree bound (using some extra properties of our orientations).

\subsection{State of the Art in Scalable MPC and Distributed Algorithms}
\paragraph{LOCAL Algorithms.} A simple LOCAL-model algorithm of Barenboim and Elkin~\cite{barenboim2008sublogarithmic} gives an orientation of edges with outdegree at most $(2+\eps)\lambda$, for any constant $\eps>0$, in $O(\log n)$ rounds: per round, simultaneously remove all nodes of degree at most $(2+\eps)\lambda$ in the remaining graph and mark their edges as outgoing (if both endpoints of an edge are removed, orient the edge arbitrarily, e.g., toward the higher identifier). This $O(\log n)$ round complexity is optimal: even when $\lambda=1$, computing any orientation with $O(1)$ outdegree requires $\Omega(\log n)$ rounds, as implied by a lower bound of Linial~\cite{linial1987LOCAL}. There are also algorithms~\cite{ghaffari2017distributed, su2020distributed} that computes $((1+\eps)\lambda)$-outdegree orientation in $\poly(\log n)$ rounds. One can easily run these LOCAL algorithms, and especially the simple $O(\log n)$ round one, in scalable MPC. However, given the power of the MPC model and the costly nature of each round in the large-scale computation world modeled by MPC, it is imperative to obtain algorithms with a significantly lower round complexity.

\paragraph{Scalable MPC Algorithms.} Ghaffari, Lattanzi, and Mitrovic~\cite[Section 4]{ghaffari2019improved} presented a scalable MPC algorithm that computes an oritentation with outdegree at most $(2+\eps)\lambda$ in $\tilde{O}(\sqrt{\log n})$ \footnote{Here and throughout, we write $f(n)=\widetilde O(g(n))$ to suppress factors polylogarithmic in $n$; that is, $f(n)=\widetilde O(g(n))$ if there exists a constant $k\ge0$ such that $f(n)=O\bigl(g(n)\,(\log n)^k\bigr)$.}  rounds.\footnote{We comment that they state this more generally for \textit{coreness decomposition}, but that's done by simply running the algorithm for every $k=(1+\eps)^{i}$ \textit{coreness/arboricity estimate} in parallel. See \cite[Exercise 2.3]{GhaffariMPCNotes} for a short sketch.} This is based on a fast simulation of the above LOCAL algorithm, using sparsification ideas in the style of the work of Ghaffari and Uitto~\cite{Ghaffari2018SparsifyingDA}. We will later give more explanations about this approach and its $\tilde{\Theta}(\sqrt{\log n})$ bottleneck. To the best of our knowledge, this $\tilde{O}(\sqrt{\log n})$ bound is state-of-the-art round complexity for scalable MPC orientation, in fact even if we allow much higher maximum outdegree bounds, e.g. $\lambda \cdot \poly(\log n)$.

There is one notable exception to the above: in the special case of $\lambda=1$---i.e., when $G$ is simply a forest---a recent scalable MPC algorithm of Grunau et al.\cite{grunau2023conditionally} computes an orientation with outdegree at most $2$ in $O(\log\log n)$ rounds. They used this (and extra properties of their orientation) to present a $3$-coloring for forests in $O(\log\log n)$ rounds. This improved on a prior $4$-coloring of Ghaffari et al.\cite{ghaffari2020improved}. However, these algorithms appear to be inherently limited to the case of forests and critically use that the local neighborhood around each node has no cycle. It is unclear how one can apply these approaches even for the case of $\lambda=2$.

\paragraph{Broader Context---LOCAL vs Scalable MPC, graph exponentiation, and sparisifcation.} There is a close connection between LOCAL algorithms and Scalable MPC algorithms. LOCAL algorithms provide a natural starting point for devising scalable MPC algorithms, and often we can hope to have much faster MPC algorithms, even exponentially faster. Let us elaborate. If we temporarily ignore the local memory constraint, a $T$-round LOCAL algorithm would imply an $O(\log T)$ round MPC algorithm, by what has been known as \textit{graph exponentiation}~\cite{lenzen08}: we have $I=O(\log T)$ iterations. In iteration $i$, we make each node learn all nodes within its distance $2^{i}$ (and their edges) by simply having each node send its $(2^{i-1})$-hop neighborhood, which it knows from the previous iteration, to all of the nodes in its $(2^{i-1})$-hop neighborhood. In the end, once (the machine holding) each node knows its $T$-hop neighborhood, it can simply simulate the LOCAL algorithm on its own. We comment that subject to some technicalities\footnote{The statement is only for randomized algorithms, it is conditioned on a (widely believed) $1$-vs-$2$ cycle conjecture, and it applies only to the component-stable MPC algorithms. See \cite{ghaffari2019conditional} for details.}, this exponential speed-up is the best that we can hope for~\cite{ghaffari2019conditional}. However, the catch is that this exponential speed-up is doable only if the local neighborhoods are small enough to fit within the memory of one machine. That is not always the case, e.g., for the $T=\Theta(\log n)$-round LOCAL orientation algorithm, the neighborhood can include all graph nodes.

Ghaffari and Uitto~\cite{Ghaffari2018SparsifyingDA} introduced a \textit{sparsification} technique for \textit{distributed algorithms}, with the aim of combatting this issue (in the context of maximal independent set and matching algorithms). In simple terms, their sparsification reduces the size of the neighborhood relevant for simulating (a part of) the LOCAL algorithm. Ghaffari et. al.\cite{ghaffari2019improved} adapted this technique to low-outdegree orientation, showing that one can ``\textit{simulate}" each phase consisting of $T'=\Theta(\sqrt{\log n})$ rounds of the $T=O(\log n)$-round LOCAL model algorithm, by a $\Theta(T')$ round LOCAL algorithm running on a randomly chosen subgraph that has maximum degree $2^{\Theta(T')}=2^{\Theta(\sqrt{\log n}})$. Hence, choosing constants right, the relevant $\Theta(T')$-hop neighborhoods have size $2^{\Theta((T')^2)} \leq n^{\delta}$ and thus fit into local memory, opening the road for graph exponentian which takes only $O(\log(T'))$ MPC rounds. Performing this repeatedly to run all the $T/T'$ phases then makes the MPC round complexity $T/T'\cdot \log(T') = \tilde{O}(\sqrt{\log n})$. This $\tilde{O}(\sqrt{\log n})$ bound is the state of the art complexity for orientation in general graphs (even for much higher outdegree bounds, e.g., $\lambda \poly(\log n)$). Breaking this round complexity is the primary objective of the present paper. 

We also comment that this complexity is also the best known for scalable MPC algorithms of maximal independent set and matching in general graphs~\cite{Ghaffari2018SparsifyingDA}, and improving that---e.g., to $\poly(\log\log n)$---is one of the central open problems in the MPC literature.

\subsection{Our Results}
We present a scalable MPC algorithm that computes an orientation with outdegree $O(\lambda \log\log n)$ in $\poly(\log\log n)$ rounds, as we formally state below.
\begin{restatable}{theorem}{MainThm}\label{thm:main}
    There is a randomized scalable MPC algorithm that given any undirected graph $G=(V, E)$ with $n=|V|$ and $m=|E|$, runs in $\poly(\log \log n)$ rounds and computes an orientation of the edges such that each node has outdegree at most $O(\lambda \log\log n)$, with high probability. Here, $\lambda$ denotes the arboricity of the graph. The algorithm uses $n^{\delta}$ memory per machine, where $\delta \in (0,1)$ is an arbitrary positive constant, and $\tilde{O}(m+n)$ words of global memory. If $\lambda$ is upper bounded by $(\log n)^{O(\log\log n)}$, the algorithm is deterministic.  
\end{restatable}

The deterministic portion of the above algorithm, which is applicable when $\lambda \leq (\log n)^{O(\log\log n)}$, actually provides more structure, which will be used for our coloring results: it partitions the set of vertices $V$ into $L=\Theta(\log n)$ layers $H_1, H_2, \dots, H_{\Theta(\log n)}$, such that each node $v\in H_i$ has at most $O(\lambda \log\log n)$ neighbors in $H_{i} \cup \dots \cup H_{L}$. The orientation is implied simply by orienting edges toward the higher layer (breaking ties arbitrarily, e.g., according to identifiers). Furthermore, the layer sizes decay exponentially and we have $|H_i| \leq n \cdot exp(-\Theta(i))$. This is what is often called $H$-partition \cite{barenboim2008sublogarithmic,grunau2023conditionally}. Using these properties, as a concrete application, we obtain a similarly fast scalable MPC algorithm for coloring graphs using $O(\lambda \log\log n)$ colors for general $\lambda$:

\begin{restatable}{theorem}{ColoringThm}\label{thm:coloring}
    There is a randomized scalable MPC algorithm that given any undirected graph $G=(V, E)$ with $n=|V|$ and $m=|E|$, runs in $\poly(\log \log n)$ rounds and computes a coloring of the vertices with $O(\lambda \log\log n)$ colors, with high probability. Here, $\lambda$ denotes the arboricity of the graph. The algorithm uses $n^{\delta}$ memory per machine, where $\delta>0$ is an arbitrary positive constant, and $\tilde{O}(m+n)$ words of global memory. 
\end{restatable}
\shortOnly{
We defer the proof of \cref{thm:coloring} to the full version of the paper (a sketch of the proof appears in the technical overview below).
}

\paragraph{Discussion.} The two results above are notable primarily in breaking the $\tilde{\Theta}(\sqrt{\log n})$ complexity barrier discussed before. However, we note that our results come with the disadvantage of increasing the outdegree bound slightly, from $O(\lambda)$ to $O(\lambda \log\log n)$. For certain settings, this can be acceptable. For instance, when the coloring resulting from the orientation is used for scheduling, this $\log\log n$ factor is only a small in time complexity. Nonetheless, obtaining the $O(\lambda)$ outdegree bound in $\poly(\log\log n)$ scalable MPC rounds remains an interesting open problem. 

\subsection{Technical Overview}
\paragraph{Orientation} For a randomized orientation algorithm, it suffices to devise an algorithm for graphs with arboricity $\lambda \leq O(\log n)$, if we are willing to lose a $(1+\eps)$ factor in the outdegree for any fixed constant $\eps>0$ which is certainly the case for us. To get such a reduction from higher arboricity graphs to these lower arboricity graphs, we randomly partition the edges of the graph into $\lambda/\Theta(\log n)$ parts. See \Cref{lem:edge_partition} for details. So, from now on in this overview, let us focus on graphs with $\lambda = O(\log n)$. 

The base idea for our orientation is to ``\textit{approximately simulate}" the natural $\Theta(\log n)$-round LOCAL process\cite{barenboim2008sublogarithmic}, but in merely $\poly(\log\log n)$ MPC rounds. Recall that the LOCAL algorithm simply consists of $L=\Theta(\log n)$ iterations, where per iteration $i$ we remove all nodes of degree $d=O(\lambda)$ in the remaining graph and we place them in a layer $H_{i}$. This gives a partition $V=H_1\sqcup H_2 \ldots \sqcup H_{L}$, which will be our reference in the rest of this overview. An orientation with outdegree at most $d$ is implied by directing toward higher layers (ties broken with IDs). 

As discussed before, the naive hope would be to speed up this $\Theta(\log n)$-round LOCAL algorithm to $\Theta(\log\log n)$ rounds in MPC, by applying the graph exponentiation described above for $\Theta(\log\log n)$ iterations and thus learning $\Theta(\log n)$-hop neighborhoods. However, this is not always possible, since the neighborhood sizes can be large and would not fit in local memory. 

On a high level, and ignoring some important aspects, our algorithm performs a variant of graph exponentiation using certain pruned views of the local neighborhoods. This allows us to control the growth of the neighborhood views for nearly all nodes of interest. However, for this to work out, per exponentiation iteration, we \textit{prune} some $O(\lambda)$ edges from each node and allow their orientation to be arbitrary. This results in our $O(\lambda \log\log n)$ outdegree overall. Let us look closer into this. 

We attempt to \fullOnly{partially} ``simulate" the first $L'=O(\delta\log n/\log\log n)$ rounds of the LOCAL algorithm. In the partition $V=H_1\sqcup H_2 \ldots \sqcup H_{L}$ fixed above (which is of course not known to our algorithms, but we use for algorithm intuition and analysis), let us drop all edges inside each layer, i.e., edges $(u,v)$ for $u,v\in H_i$ for some $i\in [1, L]$. This is only $O(\lambda)$ edges per node. Then, each remaining directed path goes through layers monotonically. For each node $v$, the total number of directed paths of length at most $L'$ starting from $v$ is at most $d^{L'} \leq n^{\delta/3}.$ Hence, with the exception of at most $n/n^{\delta/3}$ nodes, each remaining node $u$ has at most $n^{2\delta/3}$ directed paths of length at most $L'$ ending in $u$. For such a node $u$, it seems that its incoming neighborhood up to distance $L'$ would fit within the local memory of one machine, and thus one would hope that we should be able to ``perform" the part of the LOCAL model relevant to this node $u$. Of course, the catch is that this is talking about neighborhood reachable only along \textit{incoming edges} in a fixed orientation, which is unknown to the algorithm. Without this distinction, when paths are allowed to intermix outgoing and incoming edges, the number of paths reachable from one node can be very large. 

Ideally, we would have liked to ignore learning neighborhoods and exponentiation along $O(\lambda)$ \textit{outgoing edges} and let those edges be oriented arbitrarily. But we do not know which edges are outgoing. We try to do something with a similar effect: intuitively, in the hypothetical scenario that neighborhood size of $v$ could be viewed as the summation of the neighborhoods along edges to different neighbors (e.g., if the local neighborhood in consideration was simply a tree~\footnote{In a sense, one can say this is a part of the algorithm of Grunau et al.\cite{grunau2023conditionally} for orientation and coloring in forests, where they can distinguish and ``ignore" the edge/direction toward the heaviest subtree. But that heavily relies on the tree local view in their case of forests, which is far from truth in our setup of general graphs. Our algorithm is partly inspired by trying to make such an idea work in general graphs.}), we would drop up to $O(\lambda)$ edges that connect toward the largest neighborhoods. However, the major challenge is that we are dealing with general graphs: due to the existence (and potential abundance) of cycles, there is no such summation property. It is in fact plausible that because of one large region, connections to many neighbors appear to have large neighborhoods. 

Our algorithm in a sense forces the tree-like naive intuition to work out. Though, the algorithm and analysis are quite different and much more complex. We make each node maintain its neighborhood view during the exponentiation as a \textit{rooted tree}, by allowing each node to appear multiple times in the neighborhood along different branches of the tree, essentially once for every distinct path that reaches it. This gives a tree-like view of the paths connecting to the node. There is now much care needed in how these tree views (and guarantees about them) relate to the actual neighborhoods in the graph. We do attempt to summarize those aspects here, due to their technical nature. We show that during each exponentiation, we can perform a pruning on these tree-like local views, by iteratively moving from the leaves to the root and each time cutting away the $O(\lambda)$ heaviest subtrees for each node. This ensures that the non-pruned views fit into local memory in MPC and allow us to perform the exponentiation. Given the $O(\log\log n)$ iterations of exponentiation, this translates to an $O(\lambda\log\log n)$ outdegree. However, we need additional properties and arguments to show that, if we simulate the LOCAL algorithm based on these pruned tree-like views, we can translate the guarantees to the original graph.

\paragraph{Coloring.} For coloring, suppose that we have already computed the $L=\Theta(\log n)$-layer partition $V=H_1\sqcup H_2 \ldots \sqcup H_{L}$ with outdegree at most $d=O(\lambda \log\log n)$, as discussed above. Also, let us again focus on graphs with $\lambda=O(\log n)$, to which we can reduce the general case (this time by random \textit{vertex} partitioning, cf. \Cref{lem:vertex_partition}). 

From the viewpoint of the LOCAL model, we could color this graph in $\Theta(\log n)$ phases,where in each phase we do a $degree+1$ list coloring on the graph induced by $H_i$. Each phase would take $\tilde{O}(\log^{5/3}\log n)$ LOCAL rounds, thanks to the state of the art~\cite{halldorsson2022near,ghaffari2024near}. So the overall algorithm would take $O(\log n) \cdot \tilde{O}(\log^{5/3}\log n)$ rounds. 

We simulate the above much faster in MPC: Suppose edges inside each layer are duplicated as bidirectional edges, and edges across layers are directed toward the higher layer endpoint. Then, one can see that in the above algorithm, the color of each node $v\in H_i$ is impacted only by nodes reachable from $v$ along directed paths of distance at most $(L - i+1) \cdot \tilde{\Theta}(\log^{5/3}\log n)$. This allows us to do exponentiation along outgoing edges, which will be small enough to fit local memory. For instance, in one shot, we can make nodes in all the last $\Theta(\log n/\tilde{\Theta}(\log^{8/3}\log n))$ layers learn all other nodes reachable from them, in only $O(\log\log n)$ MPC rounds. For this, we perform a \textit{directed graph exponentiation} along outgoing edges. Then that can be used to locally simulate the coloring of these $\tilde{\Theta}(\log n/(\log^{8/3}\log n))$ layers. Repeating this for $\poly(\log\log n)$ iterations gives the complete coloring.

\subsection{Other Related Work} 
\paragraph{$\Delta$-dependent coloring.} Coloring algorithms have been the subject of intense study in the LOCAL model of distributed computing and we refer to the 2013 text of Barenboim and Elkin~\cite{barenboimelkin_book} for a (slightly outdated) general review. The most basic parameterization in distributed coloring is based on the maximum degree $\Delta$ in the graph. See, e.g., \cite{chang2018optimal} for the state of the art on $\Delta+1$ coloring, which runs in $\poly(\log\log n)$ rounds of the LOCAL model, and the generalizations and sharpening in \cite{halldorsson2022near,ghaffari2024near}. 

Chang et al.\cite{chang2019coloring} give a scalable MPC algorithm for $\Delta+1$ coloring, partly based on this LOCAL algorithm, which achieves a round complexity of $O(\log\log \log n)$---this needs plugging in the improved network decomposition bounds\cite{rozhonghaffari20} in their algorithm. This is the best known round complexity for scalable MPC $\Delta+1$ coloring algorithms. We also comment that in the quasilinear (and higher) memory regime of MPC, where the local memory is $\tilde{\Omega}(n)$, an algorithm of Assadi, Chen, and Khanna~\cite{assadi2019sublinear} computes $\Delta+1$ coloring in $O(1)$ rounds. However, this is vastly different, from a technical perspective, from the scalable MPC regime. 

\paragraph{Density-dependent coloring.} The $\Delta$-dependent coloring mentioned above can be too relaxed for many graphs of interest. The arboricity (or density) measure studied in this paper provides a sharper bound for coloring in a wide range of graphs. As a simple example, in a star graph, $\Delta=\Theta(n)$ and $\lambda=1$, and we would of course want the number of colors to be closer to the latter. 

If we are in the memory regime of MPC that each machine has memory $\tilde{\Omega}(n)$, one can compute a $(2+\eps)\lambda$ orientation and $(1+\eps)\lambda$ in $O(1)$ rounds via simple random partitioning and gathering the entire relevant subgraph in one machine. The problem in the scalable MPC is vastly different, however. In the LOCAL model, the best known algorithms for coloring as a function of arboricity are by Ghaffari and Lymouri\cite{ghaffariLymouri2017} and compute a coloring with $O(\lambda)$ colors in $\tilde{O}(\log n)$ rounds. See their paper for sharper statements of the bounds. See also \cite{barenboim2008sublogarithmic} for deterministic algorithms.

\paragraph{Other problems and models.} There is a plethora of other tangentially relevant work, for which providing a comprehensive review is beyond our scope. We still mention a few of the work in scalable MPC, on other graph problems, and some work on density-dependent coloring in other models. Maximal independent set and matching have been studied in low arboricity graphs\cite{behnezhad2019massively,ghaffari2020improved, fischer2023deterministic}. For general graphs (i.e., non-constant $\lambda$), ignoring $O(\log\log n)$ factors, the complexity is stuck at $\tilde{O}(\sqrt{\log \lambda})$, which is the same bottleneck of \cite{Ghaffari2018SparsifyingDA} mentioned before; it arises essentially by reducing the maximum degree $\Delta$ to $\poly(\lambda)$ and then resorting to the state-of-the-art $\Delta$-dependency of $\tilde{O}(\sqrt{\log \lambda})$. Moreover, there has been much work on scalable MPC algorithms for graph connectivity in $\tilde{O}(\log(diameter))$ rounds\cite{andoni2018parallelgraphconnectivitylog,behnezhad2019near} and its derandomizations\cite{coy2022deterministic,fischer2022improved,balliu2023optimal}. The discussions above do not emphasize the distinction between randomized and deterministic MPC algorithms. See \cite{czumaj2021graph,czumaj2021improved,coy2022deterministic,fischer2023deterministic}, and the citations therein, for examples of recent research on derandomization in MPC. See also \cite{latypov2024adaptive} and \cite{ghaffari2024dynamic,christiansen2023improved} for density-dependent coloring algorithms in other models, respectively, adaptive MPC and dynamic algorithms.

\section{Preliminaries}

The following two Lemmas are folklore. They allow to reduce the effective arboricity to $O(\log n)$, using simple random partitioning of the edges/vertices. 
\shortOnly{The proofs are deferred to the full version of the paper.}

\begin{lemma}[Edge Partitioning]
\label{lem:edge_partition}
Let $G$ be a graph with at most $n$ vertices and let $k \in \mathbb{N}$ satisfying $k \geq \lambda(G)$. Let $L := \lceil k/\log(n) \rceil$ and $G_1, G_2, \ldots, G_L$ be the $L$ graphs that one obtains by partitioning the edges into $L$ parts uniformly at random. Then, with high probability $\max_{i \in [L]} \lambda(G_i) = O(\log n)$.
\end{lemma}
\fullOnly{
\begin{proof}
Consider an orientation of the edges in $G$ with out-degree $O(\lambda(G))$. Now, consider some fixed vertex $v$. In expectation, the number of out-edges of $v$ in each of the $L$ partitions is $O(\lambda(G))/L = O(\log n)$, and a Chernoff Bound implies that this bound holds with high probability. Hence, a union boun implies that each node has $O(\log n)$ out-edges in each part with high probability, and therefore $\max_{i \in [L]} \lambda(G_i) = O(\log n)$ with high probability.
\end{proof}
}

\begin{lemma}[Vertex Partitioning]
\label{lem:vertex_partition}
Let $G$ be a graph with at most $n$ vertices and let $k \in \mathbb{N}$ satisfying $k \geq \lambda(G)$. Let $L := \lceil k/\log(n) \rceil$ and $G_1, G_2, \ldots, G_L$ be the $L$ vertex-induced subgraphs that one obtains by partitioning the vertices into $L$ parts uniformly at random. Then, with high probability $\max_{i \in [L]} \lambda(G_i) = O(\log n)$.
\end{lemma}
\fullOnly{
\begin{proof}
Consider an orientation of the edges in $G$ with out-degree $O(\lambda(G))$. Now, consider some fixed vertex $v$. In expectation, the number of out-neighbors that are in the same part as $v$ is $O(\lambda(G))/L = O(\log n)$, and a Chernoff Bound implies that this bound holds with high probability. Hence, a union boun implies that each node has $O(\log n)$ out-edges in each part with high probability, and therefore $\max_{i \in [L]} \lambda(G_i) = O(\log n)$ with high probability.
\end{proof}
}
\subsection{(Partial) Layer Assignment}
We use the following notation for the layer assignment. We note this is conceptually similar to the partition $H_1\sqcup \ldots \sqcup H_L$ discussed in the introduction, but we use the function notation here so that we have an explicit reference to the layer number of each given vertex $v$, via the related layer assignment function $\ell_G(v)$.

\begin{definition}[Partial Layer Assignment]
\label{def:partial-layer-assignment}
Let \(G\) be a graph, and let \(L\) and $d$ be positive integers. 
A \emph{partial layer assignment} of \(G\) with \(L\) layers and out-degree $d$
is a function 
\[
\ell_G \colon V(G) \to [L] \cup \{\infty\}
\]
such that for every vertex \(v \in V(G)\) with \(\ell_G(v) \neq \infty\), 
the following holds:
\[
\bigl|\{\,u \in N_G(v) : \ell_G(u) \geq \ell_G(v)\}\bigr| 
\;\le\; d.
\]
\end{definition}

In our algorithm, each node computes a partial layer assignment using only its local view. As a result, a node might receive different layer numbers from different partial assignments. We combine these different layer assignments, by simply assigning each node the smallest layer number it obtains from any of the partial assignments. The following claim shows that taking the minimum of two partial layer assignments yields another valid partial layer assignment with the same number of layers and bounded out-degree.

\begin{claim}[Min of two partial layer assignments]
\label{claim:min-of-pla}
Let \(G\) be a graph, and let \(L\) and \(d\) be positive integers.
Suppose
\[
\ell_G^{(1)} \colon V(G) \to [L] \cup \{\infty\}
\quad\text{and}\quad
\ell_G^{(2)} \colon V(G) \to [L] \cup \{\infty\}
\]
are two partial layer assignments of \(G\) with \(L\) layers and out-degree $d$.
Define \(\ell_G\) by
\[
\ell_G(v) \;=\; \min\bigl\{\ell_G^{(1)}(v), \,\ell_G^{(2)}(v)\bigr\}
\quad
\text{for all } v \in V(G).
\]
Then \(\ell_G\) is also a partial layer assignment of \(G\) with \(L\) layers and out-degree $d$.
\end{claim}
\shortOnly{
The proof of \cref{claim:min-of-pla} is deferred to the full version of the paper.

}
\fullOnly{
\begin{proof}
The function \(\ell_G\) clearly maps into \([L]\cup\{\infty\}\).
For any \(v\) with \(\ell_G(v)\neq \infty\), say \(\ell_G(v)=k\), we must show
\[
\left|\{\,u \in N_G(v) : \ell_G(u)\ge k\}\right|
\;\le\;d.
\]
Since \(k = \min(\ell_G^{(1)}(v), \ell_G^{(2)}(v)\)), it either holds $k = \ell_G^{(1)}(v)$ or $k = \ell_G^{(2)}(v)$.
If \(\ell_G^{(1)}(v)=k\), then 
\[
\{\,u \in N_G(v) : \ell_G(u)\ge k\}
\;\subseteq\;
\{\,u \in N_G(v) : \ell_G^{(1)}(u)\ge \ell_G^{(1)}(v)\},
\]
which has size at most $d$ because \(\ell_G^{(1)}\) is a partial layer assignment with out-degree $d$. 
A symmetric argument applies if \(\ell_G^{(2)}(v)=k\). Hence \(\ell_G\) also satisfies 
the out-degree condition, completing the proof.
\end{proof}
}
\begin{definition}[Strictly Increasing Paths and Path Counts]
\label{def:paths}
Let $G$ be a graph, $L$ be a positive integer, and let 
\(\ell_G \colon V(G) \to [L] \cup \{\infty\}\)
be a partial layer assignment of $G$. 

A path \(P = (v_1,v_2,\ldots,v_k)\) in \(G\) is called 
\emph{strictly increasing (with respect to \(\ell_G\))} 
if 
\[
\ell_G(v_1) \;<\; \ell_G(v_2) \;<\;\cdots\;<\; \ell_G(v_k) < \infty.
\]

For each vertex \(v \in V(G)\), define
\(\mathrm{NumPathsIn_{G, \ell_G}}(v)\) to be the number of distinct 
strictly increasing paths in \(G\) that end at \(v\), and
\(\mathrm{NumPathsOut_{G, \ell_G}}(v)\) to be the number of distinct 
strictly increasing paths in \(G\) that start at \(v\).
\end{definition}
We use the following lemma to show that $\mathrm{NumPathsIn_{G, \ell_G}}(v)$ is small for most nodes $v$. We then show that nodes with small $\mathrm{NumPathsIn_{G, \ell_G}}(v)$ will be assigned a layer by our algorithm.
\begin{lemma}
\label{lem:double_counting}
Leg $G$ be a graph, and $L$ and $d$ be positive integers. Let $\ell_G \colon V(G) \mapsto [L]$ be a complete layer assignment with out-degree $d \geq 2$. Then, it holds that

\begin{align*}
\sum_{v \in V(G)} \mathrm{NumPathsIn_{G, \ell_G}}(v) 
&= \sum_{v \in V(G)} \mathrm{NumPathsOut_{G, \ell_G}}(v) \\
&\leq |V(G)| \cdot \max_{v \in V(G)}NumPathsOut_{G, \ell_G}(v) \\
&\leq |V(G)| \cdot \left(\sum_{j=0}^{L-1} d^j\right). \\
&\leq |V(G)| \cdot d^L.
\end{align*}

\end{lemma}
\shortOnly{
The proof of \cref{lem:double_counting} is deferred to the full version of the paper.
}
\fullOnly{
\begin{proof}
The first equality follows by a simple double-counting argument. Thus, it suffices to show that for every $v \in V(G)$, $|NumPathsOut_{G, \ell_G}(v)| \leq \sum_{j=0}^{L-1} d^j$.
We prove a slightly stronger statement by induction on the layer \(i = \ell_G(v)\):
\[
\mathrm{NumPathsOut_{G, \ell_G}}(v)
\;\le\;
\sum_{j=0}^{L - i} d^j
\quad
\text{for all vertices } v \text{ with } \ell_G(v) = i.
\]
\medskip
\noindent
\textbf{Base Case (\(\ell_G(v) = L\)):}  
No vertex lies in a strictly larger layer than \(L\), so every strictly increasing path from \(v\) 
must be the single-vertex path \((v)\). Thus 
\(\mathrm{NumPathsOut_{G, \ell_G}}(v) = 1 = \sum_{j=0}^{0} d^j\), as required.

\medskip
\noindent
\textbf{Inductive Step:}
Assume the claim for all vertices in layers \(> i\). 
Let \(v\) satisfy \(\ell_G(v) = i < L\). 
Any strictly increasing path from \(v\) is either just \((v)\) 
or extends to a neighbor \(u\) with \(\ell_G(u) > i\). 
Hence,
\[
\mathrm{NumPathsOut_{G, \ell_G}}(v)
\;=\;
1 + \sum_{\substack{u \in N_G(v) \\ \ell_G(u) > i}} 
\mathrm{NumPathsOut_{G, \ell_G}}(u).
\]
Since \(v\) has at most $d$ neighbors in higher layers, and by the inductive hypothesis 
\(\mathrm{NumPathsOut_{G, \ell_G}}(u) \le \sum_{j=0}^{L - (i+1)} d^j\),
we get
\[
\mathrm{NumPathsOut_{G, \ell_G}}(v)
\;\le\;
1 
\;+\;
d
\cdot 
\sum_{j=0}^{L - i - 1} d^j
\;=\;
\sum_{j=0}^{L - i} d^j.
\]
This completes the induction, and hence the proof.
\end{proof}
}
\subsection{Valid Mappings}

During the graph exponentation procedure, each node maintains a rooted tree that captures part of its local neighborhood. Each node of the tree corresponds to a node in the original graph, and different nodes of the tree can correspond to the same node. Each edge of the tree maps to an edge in the graph and we also maintain the property that the children of a given tree node map to different nodes in the original graph.   

\begin{definition}[\textbf{Trees with Valid Mappings}]
Let $G$ be a graph, $T$ be a rooted tree, and $map: V(T) \mapsto V(G)$. We say that $map$ is a valid mapping if
\begin{enumerate}
    \item for each edge $(x,y)$ in the tree $T$, $\{map(x), map(y)\}$ is an edge in $G$, and
    \item for each node $x \in V(T)$, and any distinct two children $c_1$ and $c_2$ of $x$, it holds that $map(c_1) \neq map(c_2)$.
\end{enumerate}
\end{definition}
\begin{definition}[\textbf{Tree Pruning, with Valid Mappings}]\label{def:treePrune}
Let $G$ be a graph, $T$ be a rooted tree with a valid mapping $map\colon V(T) \mapsto V(G)$. Let $T'$ be a rooted tree, with the same root, resulting from removing some nodes in $T$. We define the corresponding valid mapping 
$\mathrm{map}'\colon V(T') \to V(G)$ 
by restricting $\mathrm{map}$ to $V(T')$. 
In other words, for every node $v \in V(T')$, we set 
$\mathrm{map}'(v) := \mathrm{map}(v)$.
\end{definition}

The tree attachement procedure will be used in each of the $O(\log \log n)$ graph exponentation steps.

\begin{definition}[\textbf{Tree Attachement, with Valid Mappings}]\label{def:treeAttachement}
Let $G$ be a graph, $T$ be a rooted tree with a valid mapping $map\colon V(T) \mapsto V(G)$ and let $x_1,x_2,\ldots, x_{\eta}$ be distinct leaves in $T$. Moreover, let $T_1,T_2,\ldots, T_{\eta}$ be rooted trees and $map_1,map_2,\ldots, map_\eta$ such that $map_i \colon V(T_i) \mapsto V(G)$ is a valid mapping, and where the root $r_i$ of $T_i$ satisfies $map(x_i) = map_i(r_i)$. We define the \textit{attachment} $T'$ to be the rooted tree resulting from $T$ by replacing each $x_i$ with (a fresh copy of) $T_i$, with the new mapping $map'$ defined as the same as $map$ for nodes of $T'$ that were in $V(T)\setminus (\cup_{i=1}^{\eta} \{x_i\})$, and the same as $map_i$ for nodes in $T_i$ that are now in $T'$.
\end{definition}

\begin{definition}[\textbf{Missing Neighbors}]
\label{def:missing}
    Let $G$ be a graph, $T$ be a rooted tree and $map \colon V(T) \mapsto V(G)$ be a valid mapping. For a node $x \in V(T)$, we define 
    \[Missing_{G,T,map}(x) := N_G(map(x)) \setminus \{map(c) \colon c \in children_T(x) \}.\]
\end{definition}

\begin{definition}[\textbf{Strictly Monotonically Reachable}]
\label{def:strictly monotonic reachable}
Let \(G\) be a graph and let \(\ell_G \colon V(G) \to \mathbb{N} \cup \{\infty\}\) be a partial layer assignment.  
Further, let \(T\) be a rooted tree with root \(r\), and let \(\mathit{map} \colon V(T) \to V(G)\) be a valid mapping.  
A vertex \(x \in V(T)\) is said to be \emph{strictly monotonically reachable} with respect to \(\ell_G\) if, 
when we denote by 
\[
  x = x_1, \, x_2, \, \ldots, \, x_k = r
\]
the unique path from \(x\) to the root \(r\) in \(T\), it holds that
\[
  \ell_G\bigl(\mathit{map}(x_1)\bigr) 
  \;<\; 
  \ell_G\bigl(\mathit{map}(x_2)\bigr) 
  \;<\; 
  \cdots 
  \;<\; 
  \ell_G\bigl(\mathit{map}(x_k)\bigr).
\]
\end{definition}

\section{Edge-Orientation Algorithm}
We are now ready to explain the different parts of our edge-orientation algorithm.
\subsection{Prune}
The procedure LocalPrune recursively removes the heaviest $k$ subtrees. We will use it for $k = O(\lambda(G))$. The algorithm will be executed locally on a single machine without any communication.

\fullOnly{\begin{algorithm}[H]}
\shortOnly{\begin{algorithm}}
\SetAlgoLined
\KwIn{A rooted tree $T$ with root $r$, a pruning parameter $k$}
\KwOut{A rooted subtree $T_{\text{pruned}}$ of $T$}

\If{$r$ has at most $k$ children}{
\Return the single-node rooted tree consisting only of $r$
}
\ForEach{child $c$ of $r$}{
    Let $T_c$ be the subtree rooted at $c$
    
    $T_{c, \text{pruned}} \gets \text{LocalPrune}(T_c, k)$
}
\BlankLine
\tcp{ Identify and remove the largest $k$ pruned subtrees among the children}
Let $\mathcal{C} \gets \{\,T_{c,\mathrm{pruned}} : c \text{ is a child of }r\}\,$

Sort the subtrees in $\mathcal{C}$ by size in \emph{descending} order (ties broken arbitrarily)

Remove the first $k$ subtrees from $\mathcal{C}$ (i.e., the $k$ largest)

\BlankLine
\tcp{ Form the pruned tree by attaching the remaining children to $r$}
Create a new tree $T_{\mathrm{pruned}}$ with root $r$

\ForEach{$T_{c,\mathrm{pruned}} \in \mathcal{C}$}{
    Attach $T_{c,\mathrm{pruned}}$ as a subtree of $r$ in $T_{\mathrm{pruned}}$
}

\Return $T_{\mathrm{pruned}}$
\caption{LocalPrune}
\label{alg:local_prune}
\end{algorithm}

\begin{claim}
\label{claim:pruned_missing}
Let $G$ be a graph, $T$ be a rooted tree, and let $map \colon V(T) \mapsto V(G)$ be a valid mapping.  
Fix any integer $k \ge 1$, and  consider $T_{\mathrm{pruned}} \gets LocalPrune(T,k)$ with the induced mapping
\[
  map_{\mathrm{pruned}}(x) \;:=\; map(x) 
  \quad\text{for every } x \in V\bigl(T_{\mathrm{pruned}}\bigr).
\]
Then, for every $x \in V\bigl(T_{\mathrm{pruned}}\bigr)$, it holds that
\[
  Missing_{G,T_{\mathrm{pruned}},map_{\mathrm{pruned}}}(x) 
  \;\;\le\;\; Missing_{G,T,map}(x) \;+\; k.
\]
\end{claim}
\shortOnly{The proof of \cref{claim:pruned_missing} is deferred to the full version of the paper.}
\fullOnly{
\begin{proof} We present a proof by induction on the height of \(T\): 

\textbf{Base case.}
If the root \(r\) of \(T\) has at most \(k\) children, then the pruning algorithm returns the single-node tree just consisting of $r$. As $r$ looses at most $k$ children, it trivially holds that 
\[
  Missing_{G,T_{\mathrm{pruned}},map_{\mathrm{pruned}}}(r) 
  \;\;\le\;\; Missing_{G,T,map}(r) \;+\; k.
\] 
No other nodes remain, so the claim trivially holds.

\smallskip
\textbf{Inductive step.}
Suppose \(r\) has more than \(k\) children \(c_1,\dots,c_m\).  
For each \(c_j\), consider the subtree \(T_{c_j}\) rooted at \(c_j\) and let \(T_{c_j, pruned}\) be its pruned version, and $map_{c_j, pruned} \colon V(T_{c_j, pruned}) \mapsto V(G)$ its mapping.  
By the inductive hypothesis, every node \(x\) in \(T_{c_j, pruned}\) satisfies
\[
  \text{Missing}_{G,T_{c_j,pruned},\,\mathrm{map}_{c_j, pruned}}(x)
  \;\le\;
  \text{Missing}_{G,T_{c_j},\,\mathrm{map}}(x)
  \;+\; k.
\]
The algorithm then removes the largest \(k\) pruned subtrees from \(r\)'s children.  
Hence, as each $T_{c_j, pruned} $ contains at least one node, \(r\) loses exactly \(k\) children, causing its missing set to increase by at most \(k\). 

Thus, for every \(x\) in the final pruned tree \(T_{\mathrm{pruned}}\), 
\(\text{Missing}_{G,T_{\mathrm{pruned}},\,\mathrm{map}_{pruned}}(x)\)
is at most
\(\text{Missing}_{G,T,\,\mathrm{map}}(x) + k\).
\end{proof}
}

\begin{lemma}\label{lem:local-prune-size}
Let $G$ be a graph, and let $T$ be a rooted tree with root $r$. 
Suppose we have a \emph{valid} mapping 
\[
\mathrm{map} \colon V(T) \,\to\, V(G).
\]
Let $L$ and $d$ be positive integers, and let 
\(\ell_G \colon V(G)\to [L]\cup\{\infty\}\) 
be a partial layer assignment of $G$ with out-degree $d$ satisfying that $\ell_G(map(r)) \neq \infty$. 
Furthermore, let $k \geq d$ be a positive integer, and 
obtain $T_{\mathrm{pruned}}$ by calling 
\(\textsc{LocalPrune}(T, k)\). 
Then,
\[
\bigl|\,V\bigl(T_{\mathrm{pruned}}\bigr)\bigr|
\;\le\;
\mathrm{NumPathsIn_{G, \ell_G}}\!\bigl(\mathrm{map}(r)\bigr).
\]
\end{lemma}
\shortOnly{
The proof of \cref{lem:local-prune-size} is deferred to the full version of the paper.
} 
\fullOnly{
\begin{proof} We present a proof by an induction on $\ell_G(\mathrm{map}(r))$:

\textbf{Base Case: \(\ell_G(\mathrm{map}(r)) = 1\).} Since \(\ell_G\) has out-degree $d$, any vertex in layer \(1\) can have at most 
$d$ neighbors in $G$.  
Because \(\mathrm{map}\) is valid, this implies \(r\) has at most \(d \le k\) children in \(T\). 
Hence, the algorithm \(\textsc{LocalPrune}(T,k)\) prunes \(T\) to the single-node tree \(\{r\}\). 
Thus, \(\bigl|V(T_{\mathrm{pruned}})\bigr|\le 1\). 
On the other hand, there is at least one strictly increasing path ending at \(\mathrm{map}(r)\) (the trivial path), 
so 
\(\mathrm{NumPathsIn_{G, \ell_G}}\!\bigl(\mathrm{map}(r)\bigr)\ge1\). 
Therefore,
\[
\bigl|V\bigl(T_{\mathrm{pruned}}\bigr)\bigr|
\;\le\;
1
\;\le\;
\mathrm{NumPathsIn_{G, \ell_G}}\!\bigl(\mathrm{map}(r)\bigr),
\]
which settles the base case.

\textbf{Inductive Step: \(\ell_G(\mathrm{map}(r)) = j > 1\).}
Assume the lemma holds for all rooted trees \(T'\) whose root \(r'\) satisfies 
\(\ell_G(\mathrm{map}(r')) < j\).  
We prove it for our current root \(r\) with \(\ell_G(\mathrm{map}(r)) = j\).  

\begin{enumerate}
\item
\textbf{Split the children of \(r\) by layer.}
Partition the children of \(r\) into two sets:
\[
C_{< j} 
\;=\;
\bigl\{\, c \in \mathrm{children}_T(r) : \ell_G(\mathrm{map}(c)) < j\bigr\}
\quad\text{and}\quad
C_{\ge j}
\;=\;
\mathrm{children}_T(r) \,\setminus\, C_{< j}.
\]
By the partial layer assignment’s out-degree condition (and since $map$ is valid and \(\ell_G(\mathrm{map}(r))\neq\infty\)), 
the vertex \(\mathrm{map}(r)\) can have at most \(d\) neighbors in the same or a \emph{larger} layer;  
in particular,
\[
\bigl|C_{\ge j}\bigr|
\;\;\le\;\;
d
\;\;\le\;\;
k.
\]

\item
\textbf{Recursively prune each child.}
For every child \(c\), let \(T_c\) be the subtree rooted at \(c\).  
The algorithm calls \(\textsc{LocalPrune}(T_c, k)\) to produce \(T_{c,\mathrm{pruned}}\).  

\begin{itemize}
\item
If \(c \in C_{< j}\), then \(\ell_G(\mathrm{map}(c)) < j\).  
By the inductive hypothesis, 
\[
\bigl|V\bigl(T_{c,\mathrm{pruned}}\bigr)\bigr|
\;\;\le\;\;
\mathrm{NumPathsIn_{G,\ell_G}}\!\bigl(\mathrm{map}(c)\bigr).
\] 
\end{itemize}

\item
\textbf{Removing the largest \(k\) pruned subtrees.}  
After pruning each child’s subtree, \(\textsc{LocalPrune}(T,k)\) sorts these pruned subtrees 
by size (in descending order) and removes the first \(k\) from consideration.  
Denote by \(\mathcal{R}\) the set of indices (children) whose subtrees remain.  
Thus, 
\[
\bigl|V\bigl(T_{\mathrm{pruned}}\bigr)\bigr|
\;=\;
1
\;+\;
\sum_{c \in \mathcal{R}}
\bigl|V\bigl(T_{c,\mathrm{pruned}}\bigr)\bigr|.
\]

\item
\textbf{Bounding the size of \(T_{\mathrm{pruned}}\).}
Because $|C_{\geq j}| \leq k$ and the $k$ largest pruned subtrees are removed, we get
\[\bigl|V\bigl(T_{\mathrm{pruned}}\bigr)\bigr|
\;=\;
1
\;+\;
\sum_{c \in \mathcal{R}}
\bigl|V\bigl(T_{c,\mathrm{pruned}}\bigr)\bigr| \leq 1 + \sum_{c \in C_{< j}} |V(T_{c,pruned})|.
\]
Hence,

\begin{align*}
    \bigl|V\bigl(T_{\mathrm{pruned}}\bigr)\bigr| &\leq 1 + \sum_{c \in C_{< j}} |V(T_{c,pruned})| \\
    &\leq  1 + \sum_{c \in C_{< j}} \mathrm{NumPathsIn_{G,\ell_G}}\!\bigl(\mathrm{map}(c)\bigr) \\
    &\leq 1 + \sum_{u \in N_G(map(r)) \colon \ell_G(u) < \ell_G(map(r))} \mathrm{NumPathsIn_{G,\ell_G}}\!\bigl(u\bigr)  \\
    &= \mathrm{NumPathsIn_{G,\ell_G}}\!\bigl(map(r)\bigr). 
\end{align*}
\end{enumerate}
\end{proof}
}

\subsection{Exponentiate + Prune}
We now describe our graph exponentiation algorithm, which runs in $s = O(\log \log n)$
steps. In this algorithm, every node \(v\) maintains a rooted tree along with a valid mapping that assigns the tree's root to \(v\). Crucially, the size of each rooted tree is kept within a budget \(B\).

Each of the \(s\) steps proceeds in two phases:
\begin{enumerate}
    \item \textbf{Local Pruning:} Each node prunes its rooted tree using the procedure detailed in \cref{alg:local_prune}. This step is entirely local (each tree is stored on a single machine) and thus requires no communication.
    \item \textbf{Graph Exponentiation:} Using the pruned trees, the algorithm then performs a graph exponentiation step via the attachment procedure defined in \cref{def:treeAttachement}.
\end{enumerate}

\fullOnly{\begin{algorithm}[H]}
\shortOnly{\begin{algorithm}}
\SetAlgoLined
\KwIn{A graph $G$, a budget $B$, a pruning parameter $k$, and a step count $s$}
\fullOnly{
\KwOut{\mbox{For each $v \in V(G)$, a rooted subtree $T_v$ with a valid mapping $map_v \colon V(T_v)\to V(G)$.}}
}
\shortOnly{
\KwOut{\mbox{$\forall v \in V(G)$, a rooted tree $T_v$ with a mapping $map_v$.}}
}

\BlankLine
\textbf{Initialization:}\\
\ForEach{$v\in V(G)$ with $|N_G(v)| < B$}{
  \textbf{Construct} a rooted tree $T^{(0)}_v$ with $|N_G(v)|$ children.\\
  \textbf{Define} $map^{(0)}_v \colon V(T_v^{(0)}) \mapsto V(G)$ so that the root of $T^{(0)}_v$ maps to $v$ and each child maps to a distinct neighbor in $N_G(v)$.\\
  \textbf{Mark} $v$ as \emph{active}.
}
\ForEach{$v\in V(G)$ with $|N_G(v)| \geq B$}{
  \textbf{Construct} a rooted tree $T^{(0)}_v$ with a single node that maps to $v$.\\
  \textbf{Mark} $v$ as \emph{inactive}.
}

\BlankLine
\For(\tcp*[h]{Main Loop}){$i=1$ \KwTo $s$}{
  \BlankLine
  \textbf{Local Prune Step:}\\
  \ForEach{$v\in V(G)$}{
    $T^{(i-1)}_{v,\mathrm{pruned}} \gets \text{LocalPrune}\bigl(T^{(i-1)}_v, k\bigr)$.\\
    \textbf{Restrict} $map^{(i-1)}_v$ to obtain $map^{(i-1)}_{v,\mathrm{pruned}}$ on $T^{(i-1)}_{v,\mathrm{pruned}}$.\\[4pt]
    \If(\tcp*[h]{Check size vs.\ budget $B$}){$\bigl|V\bigl(T^{(i-1)}_{v,\mathrm{pruned}}\bigr)\bigr| > \sqrt{B}$}{
      \textbf{Mark} $v$ as \emph{inactive}.
    }
  }

  \BlankLine
  \textbf{Exponentiation / Attachment Step:}\\
  \ForEach{$v \in V(G)$}{
    \uIf{$v$ is \emph{inactive}}{
      $T^{(i)}_v \gets T^{(i-1)}_{v,\mathrm{pruned}}$ \tcp*[f]{No further expansion}
    }
    \Else{
    Let $x_1,\dots,x_\eta$ be the leaves at distance exactly $2^{i-1}$ from the root of $T^{(i-1)}_{v,\mathrm{pruned}}$ that map to an active vertex.\\
\ForEach{$x_j$}{
  Let $u_j \gets map^{(i-1)}_{v,\mathrm{pruned}}(x_j)$ and retrieve the subtree 
  $T^{(i-1)}_{u_j,\mathrm{pruned}}$ together with its valid mapping 
  $map^{(i-1)}_{u_j,\mathrm{pruned}}$. 
}
\textbf{Attach} these pruned subtrees 
$\bigl(T^{(i-1)}_{u_j,\mathrm{pruned}},\,map^{(i-1)}_{u_j,\mathrm{pruned}}\bigr)$ 
at the leaves $x_j$ of $T^{(i-1)}_{v,\mathrm{pruned}}$ (see Def.~\ref{def:treeAttachement}).\\
\textbf{Obtain} $T^{(i)}_v$ and its valid mapping $map^{(i)}_v$.
    }
  }
}

\BlankLine
\textbf{Return} $T^{(s)}_{v}$ and $map^{(s)}_{v}$ for every $v \in V(G)$.

\caption{ExponentiateAndLocalPrune}
\label{alg:exponentiatepluslocalprune}

\end{algorithm}

\begin{claim}
For every $i \in \{0,1,\ldots,s\}$ and $v \in V(G)$, it holds that $map^{(i)}_v \colon V(T_v^{(i)}) \mapsto V(G)$ is a valid mapping.
\end{claim}

\begin{proof}
Recall that a mapping $map \colon V(T)\to V(G)$ is \emph{valid} if:
\begin{enumerate}
  \item For each edge $(x,y)$ in $T$, the pair $\{map(x), map(y)\}$ is an edge in $G$.
  \item For every node $x \in V(T)$, any two distinct children $c_1 \neq c_2$ of $x$ satisfy $map(c_1) \neq map(c_2)$.
\end{enumerate}
We prove the claim by induction on $i$.

\paragraph{Base Case ($i=0$).}
If $|N_G(v)| < B$, then $T^{(0)}_v$ is constructed so that its root is mapped to $v$, and the children of that root are mapped to the distinct neighbors of $v$ in $G$. Thus, every tree edge $(\text{root}, \text{child})$ is mapped to an edge in $G$, and no two children of the root share the same image. Hence $map^{(0)}_v$ is valid for all $v \in V(G)$. If $|N_G(v)| \geq B$, then $T^{(0)}_v$ consists just of a root and is therefore trivially valid.

\paragraph{Inductive Step ($i-1 \to i$).}
Assume that for some $i-1 \ge 0$, the mapping $map^{(i-1)}_u$ is valid for every $u \in V(G)$. We show $map^{(i)}_v$ is valid for each $v \in V(G)$:

\begin{itemize}
    \item \textbf{Pruning.}
    We obtain $T^{(i-1)}_{v,\mathrm{pruned}}$ by removing some nodes from $T^{(i-1)}_v$.  Since $map^{(i-1)}_v$ was valid on $T^{(i-1)}_v$, restricting it to the subtree $T^{(i-1)}_{v,\mathrm{pruned}}$ preserves adjacency of images and the distinctness condition on siblings.

    \item \textbf{Attachment.}  
    By Definition~\ref{def:treeAttachement}, each leaf $x_j$ of $T^{(i-1)}_{v,\mathrm{pruned}}$ (mapped to $u_j = map^{(i-1)}_{v,\mathrm{pruned}}(x_j)$) is \emph{replaced} by the entire subtree $T^{(i-1)}_{u_j,\mathrm{pruned}}$. Concretely,
    \begin{enumerate}
        \item Remove $x_j$ from $T^{(i-1)}_{v,\mathrm{pruned}}$.
        \item Insert $T^{(i-1)}_{u_j,\mathrm{pruned}}$, whose root is also mapped to $u_j$, so its parent edge in the new tree connects the parent of $x_j$ to the root of the inserted subtree.
    \end{enumerate}
    The inductive hypothesis ensures each inserted subtree already has a valid mapping.  Since the parent of $x_j$ was mapped to some node $p$ in $G$, and the new root is also mapped to $u_j$, adjacency is preserved (because $\{p,u_j\}$ is an edge in $G$ when $x_j$ existed).  Moreover, children of different inserted subtrees or siblings from within each valid subtree remain mapped to distinct vertices.
\end{itemize}

Therefore, $T^{(i)}_v$ is mapped validly into $G$, concluding the induction.

\paragraph{Conclusion.}
By induction on $i$, $map^{(i)}_v$ is valid for all $i \in \{0,\ldots,s\}$ and $v \in V(G)$.
\end{proof}

\shortOnly{
The following claim follows by a simple induction.
}
\begin{claim}
\label{claim:tree_sizes}
For every $v \in V(G)$ and every integer $i \in \{0,1,\ldots,s\}$, the tree $T^{(i)}_v$ has at most $B$ nodes.
\end{claim}
\fullOnly{
\begin{proof}[Proof (by induction on $i$)]
\noindent
\textbf{Base Case ($i=0$).}
If $|N_G(v)| \geq B$, then $T_v^{(0)}$ just consist of a single node. If $ |N_G(v)| < B$, then $T^{(0)}_v$ is a star with $1 + |N_G(v)| \leq B$ nodes.

\medskip

\noindent
\textbf{Inductive Step ($i \to i+1$).}
Assume that for some $i \in  \{0, 1,\ldots,s -1\}$, for every vertex $u \in V(G)$ we have  $\bigl|T^{(i)}_u\bigr| \;\le\; B$.
We must show that $|T^{(i+1)}_v| \le B$ for each $v$.

\smallskip

\emph{Case 1: $v$ is inactive.}
Then $T^{(i+1)}_v = T^{(i)}_{v,\mathrm{pruned}} \subseteq T^{(i)}_v$, so
\[
  \bigl|T^{(i+1)}_v\bigr| \;\le\; \bigl|T^{(i)}_v\bigr|
  \;\le\; B
  \quad (\text{by the induction hypothesis}).
\]

\smallskip

\emph{Case 2: $v$ is active.}
Here $|T^{(i)}_{v,\mathrm{pruned}}| \le \sqrt{B}$, and every attached subtree $T^{(i)}_{u_j,\mathrm{pruned}}$ also has size $\le \sqrt{B}$.  Since there can be at most $\sqrt{B}$ such attached subtrees, we get
\[
  \bigl|T^{(i+1)}_v\bigr|
  \;\le\;
  \bigl|T^{(i)}_{v,\mathrm{pruned}}\bigr|
  \cdot
  \max(1,\max_j \bigl|T^{(i)}_{u_j,\mathrm{pruned}}\bigr| )\le \sqrt{B}\cdot\sqrt{B}
  = B.
\]
\end{proof}
}
\begin{claim}
\cref{alg:exponentiatepluslocalprune} can be implemented in $O(s)$ MPC rounds with $O(n^\delta + B)$ words of local memory and $O(nB+m)$ words of global memory, where $n$ is an upper bound on the number of vertices, and $m$ is an upper bound on the number of edges of the graph $G$.
\label{claim:MPC_implementation}
\end{claim}
\begin{proof}
This claim directly follows from \cref{claim:tree_sizes} together with a straightforward implementation using standard MPC primitives developed in previous works, e.g. [\cite{andoni2018parallelgraphconnectivitylog}, Section E] and the references therein.
\end{proof}

\shortOnly{The claim below follows by a simple induction. The full proof is deferred to the full version of the paper.}
\begin{claim}
\label{claim:bound_on_missing}
At iteration $i$, for every node $x \in T_v^{(i)}$ with $dist_{T_v^{(i)}}(r,x) < 2^i$ that maps to an active node, we have

\[|Missing_{G,T^{(i)}_v, map^{(i)}_v}(x)| \leq ik.\]

\end{claim}
\fullOnly{
\begin{proof}

\textbf{Base Case ($i = 0$):}

When $i = 0$, the root is the only node of distance strictly less than $2^{0} = 1$ from the root. If the root maps to an active node, then since pruning has not yet occurred, the root node has all its neighbors mapped in the initial tree $T^{(0)}_v$.  
  Therefore, the number of missing neighbors for the root node is 
    \[
      |\mathrm{Missing}_{G,T_v^{(0)}, map^{(0)}_v}(r)| = 0 \leq 0 \cdot k.
    \]
  
Thus, the base case holds.

\textbf{Inductive Step:}

Consider a node $x$ at a distance $\mathrm{dist}_{T^{(i+1)}_v}(r, x) < 2^{i+1}$ from the root $r$ in the tree $T^{(i+1)}_v$ that maps to an active node. We have to show that 

\[|Missing_{G,T^{(i+1)}_v, map^{(i+1)}_v}(x)| \leq (i+1)k.\]

There are two possibilities based on the distance of $x$ to the root $r$ in the tree after the $(i+1)$-th iteration:

\textbf{Case 1: $\mathrm{dist}_{T^{(i+1)}_v}(r, x) < 2^{i}$}

   Since $\mathrm{dist}_{T^{(i+1)}_v}(r, x) < 2^{i}$, by the induction hypothesis, node $x$ has at most $i \cdot k$ missing neighbors in $T^{(i)}_v$.
   During the $(i+1)$-th iteration, the \textsc{LocalPrune} step may remove up to $k$ subtrees from each node, potentially increasing the number of missing neighbors for $x$ by at most $k$ (see \cref{claim:pruned_missing}).
   Therefore, in $T^{(i+1)}_v$, node $x$ satisfies:
     \[
       |Missing_{G,T^{(i+1)}_v, map^{(i+1)}_v}(x)| = |Missing_{G, T^{(i)}_{v, pruned},map^{(i)}_{v, pruned}}|  \leq |Missing_{G,T^{(i)}_v, map^{(i)}_v}(x)| + k \leq (i+1) \cdot k.
     \]

\textbf{Case 2: $2^{i} \leq \mathrm{dist}_{T^{(i+1)}_v}(r, x) < 2^{i+1}$}
Node $x$ lies in a subtree that was attached to a leaf at distance exactly $2^{i}$ from the root during the $(i+1)$-th iteration. Let $r'$ denote the root of that subtree and $u' \gets map^{(i+1)}_v(r')$ the node that $r'$ maps to. It holds that 

\[dist_{T^{(i)}_{u'}}(r',x) = dist_{T^
{(i+1)}_v}(r,x) - dist_{T^
{(i+1)}_v}(r,r') = dist_{T^
{(i+1)}_v}(r,x) - 2^i < 2^{i+1}- 2^{i} = 2^i.\]

As $dist_{T^{(i)}_{u'}}(r',x) < 2^i$ and $x$ maps to an active node, we can apply the induction hypothesis to conclude that

\[|Missing_{G,T^{(i)}_{u'}, map^{(i)}_{u'}}(x)| \leq i \cdot k.\]

 Hence, we can use \cref{claim:pruned_missing} to conclude that

 \begin{align*}
  |Missing_{G,T^{(i+1)}_v, map^{(i+1)}_v}(x)| 
  &= |Missing_{G,T^{(i)}_{u',pruned}, map^{(i)}_{u',pruned}}(x)|\\ 
  &\leq |Missing_{G,T^{(i)}_{u'}, map^{(i)}_{u'}}(x)| + k \\
  &\leq (i+1)k.
 \end{align*}

\end{proof}
}
\begin{lemma}
\label{lem:missing_upper_bound}
Let $G$ be a graph, $d, L \in \mathbb{N}$ and $\ell_G \colon V(G) \mapsto [L] \cup \{\infty\}$ such that for every $v \in V(G)$ with $\ell_G(v) \neq \infty$, $|\{u \in N_G(v) \colon \ell_G(u) \geq \ell_G(v)\}| \leq d$.

Furthermore, let $B,k,s$ be three non-negative integers satisfying $k \geq d$ and $s > \log_2(L)$.

Let $v \in V(G)$ with $\ell_G(v) \neq \infty$ and
$NumPathsIn_{G,\ell_G}(v) \leq \sqrt{B}.$

Let $T_v^{(s)}$ and $map_v^{(s)}$  be the rooted tree and mapping that one obtains by calling

$ExponentiateAndLocalPrune(G, B, k, s)$.

Then, for every $x \in V(T_v^{(s)})$ that is strictly monotonically reachable with respect to $\ell_G$ (see \cref{def:strictly monotonic reachable}), it holds that 
\[|Missing_{G,T_v^{(s)}, map_v^{(s)}}(x)| \leq s \cdot k.\]
\end{lemma}
\begin{proof}
Consider any node $x \in T_v^{(s)}$ that is strictly monotonically reachable with respect to $\ell_G$. As $\ell_G(v) \leq L$ and $x$ is strictly monotonically reachable, it follows that the distance of $x$ to the root is at most $L - 1 < 2^{s}$. Moreover, $x$ being monotonic reachable also implies that 
\[NumPathsIn_{G,\ell_G}(map(x)) \leq NumPathsIn_{G,\ell_G}(v) \leq \sqrt{B}.\]
Hence, as $k \geq d$, \cref{lem:local-prune-size} implies that $V(T^{(i)}_{map(x), pruned}) \leq NumPathsIn_{G,\ell_G}(map(x)) \leq \sqrt{B}$ for every $i \in \{0,1,\ldots,s-1\}$ and therefore $map(x)$ is active until the end. Hence, $x$ has a distance of less than $2^s$ to the root and it maps to an active node. Therefore, we can use \cref{claim:bound_on_missing} to conclude that 

\[|Missing_{G,T^{(s)}_v, map^{(s)}_v}(x)| \leq s \cdot k.\]
\end{proof}

\subsection{Partial Layer Assignment}

After executing the graph exponentation algorithm (\cref{alg:exponentiatepluslocalprune}), we have one rooted tree for each node $v$. For each tree, we use \cref{alg:partiallayerassignmentree} to assign a layer to a subset of the vertices of the tree using a simple peeling process, where we later set $a = O(\lambda(G)\log \log n)$. \cref{alg:partiallayerassignmentree} will be executed on a single machine and does not require any communication (assuming that each tree node knows $|Missing(x)|$).

\fullOnly{\begin{algorithm}[H]}
\shortOnly{\begin{algorithm}}
\SetAlgoLined
\KwIn{Graph $G$, rooted tree $T$, valid mapping $map \colon V(T) \mapsto V(G)$, $a \in \mathbb{N}$, $L  \in \mathbb{N}$}
\KwOut{A layer assignment $\ell_T : V(T) \mapsto [L] \cup \{\infty\}$}
For every $x \in V(T)$, let $Missing(x)
 := N_G(map(x)) \setminus \{map(c) \colon \text{$c \in children_T(x)$}\}$ 
$V_{\geq 1} \gets V(T)$

\For{$j= 1, 2, \ldots, L$}{

$V_j \gets \left\{x \in V_{\geq j} \colon |children_T(x) \cap V_{\geq j} | + |Missing(x)| \leq a\right\}$ 

For every $x \in V_j, \ell_T(x) \gets j$

$V_{\geq j+1} \gets V_{\geq j} \setminus V_j$
}
For every $x \in V_{\geq L + 1}, \ell_T(x) \gets \infty$ 

\Return $\ell_T$
\caption{PartialLayerAssignmentTree}
\label{alg:partiallayerassignmentree}
\end{algorithm}

\begin{lemma}
\label{lem:layer_assignment_tree}
Let $G$ be a graph, $d, L \in \mathbb{N}$ and $\ell_G \colon V(G) \mapsto [L] \cup \infty$ such that for every $v \in V(G)$ with $\ell_G(v) \neq \infty$, $|\{u \in N_G(v) \colon \ell_G(u) \geq \ell_G(v)\}| \leq d$. 
Moreover, let $T$ be a rooted tree, $map \colon V(T) \mapsto V(G)$ be a valid mapping and $missing \in \mathbb{N}_0$ be a non-negative integer such that the following holds: for every $x \in V(T)$ that is strictly monotonically reachable with respect to $\ell_G$ (see \cref{def:strictly monotonic reachable}), it holds that $|Missing(x)| \leq missing$ (where $Missing(x) := N_G(map(x)) \setminus \{map(c) \colon \text{$c \in children_T(x)$}$).
Consider $a \in \mathbb{N}$ with $a \geq d + missing$ and let $\ell_T \colon V(T) \mapsto [L] \cup \{\infty\}$ be the \textit{layer assignment} one obtains from calling $PartialLayerAssignmentTree(G, T, map, a, L)$. Then, for every $x \in V(T)$ that is strictly monotonically reachable with respect to $\ell_G$, we have $\ell_T(x) \leq \ell_G(map(x))$.
\end{lemma}
\shortOnly{The proof of \cref{lem:layer_assignment_tree} is deferred to the full version of the paper.}
\fullOnly{
\begin{proof}
We prove the following statement by induction on $\ell_G(\mathrm{map}(x))$:
\[
  \text{For every } x \in V(T) 
    \text{ that is strictly monotonically reachable with respect to $\ell_G$}, 
    \quad 
  \ell_T(x) \leq \ell_G\bigl(\mathrm{map}(x)\bigr).
\]

\paragraph{Base Case: $\ell_G(\mathrm{map}(x)) = 1$.}
Let $x \in V(T)$ be strictly monotonically reachable with respect to $\ell_G$ and $\ell_G(\mathrm{map}(x)) = 1$. 
We must show that $\ell_T(x) = 1$.  
By definition of $V_1$ in \textsc{PartialLayerAssignment}, it suffices to prove
\[
  \bigl|\mathrm{children}_T(x)\bigr|
  \;+\;
  \bigl|\texttt{Missing}(x)\bigr|
  \;\le\;
  a,
\]
where 
\[
  \texttt{Missing}(x) 
    := 
  N_G\bigl(\mathrm{map}(x)\bigr) \setminus 
  \{\mathrm{map}(c) : c \in \mathrm{children}_T(x)\}.
\]
Since $x$ is striclty monotonic reachable, the lemma’s assumption implies 
\(\bigl|\texttt{Missing}(x)\bigr|\leq \textit{missing}\).  

Because $\mathrm{map}$ is a valid mapping, it holds that

\[|children_T(x)| \leq |N_G(map(x))|.\]

Also, from $\ell_G(\mathrm{map}(x))=1$ and the property of $\ell_G$ we have:
\[\bigl|N_G(map(x))\bigr| =
  \bigl|\{\,u \in N_G(\mathrm{map}(x)) : \ell_G(u)\,\ge\,1\}\bigr|
  \;\le\;
  d.
\]
Hence $\bigl|\mathrm{children}_T(x)\bigr| 
\leq 
\bigl|\{\,u \in N_G(\mathrm{map}(x)) : \ell_G(u)\ge 1\}\bigr|
\leq d.$  

Combining these two bounds and using $a \geq \textit{outdegree} + \textit{missing}$, we get
\[
  \bigl|\mathrm{children}_T(x)\bigr|
  \;+\;
  \bigl|\texttt{Missing}(x)\bigr|
  \;\le\;
  d + \textit{missing}
  \;\le\;
  a.
\]
Thus $x \in V_1$, which means $\ell_T(x) = 1$ as required.

\paragraph{Inductive Step.}
Let $j \in [L-1]$. Suppose that for all $y \in V(T)$ that are strictly monotonically reachable and 
\(\ell_G(\mathrm{map}(y)) \le j\), we already know 
\(\ell_T(y)\le \ell_G(\mathrm{map}(y))\).  

Now consider a node $x \in V(T)$ that is 
\[\text{strictly monotonically reachable and }
 \text{and}\quad
  \ell_G\bigl(\mathrm{map}(x)\bigr) \;=\; j+1.
\]
We want to show $\ell_T(x)\le j+1$.  
From the definition of $V_{j+1}$, it is enough to check:
\[
  \bigl|\mathrm{children}_T(x)\,\cap\,V_{\ge j+1}\bigr|
  \;+\;
  \bigl|\texttt{Missing}(x)\bigr|
  \;\le\;
  a.
\]
First, $x$ being strictly monotonically reachable implies $\bigl|\texttt{Missing}(x)\bigr|\le \textit{missing}$.  
Hence, by $a\ge d + \textit{missing}$, it suffices to show
\[
  \bigl|\mathrm{children}_T(x)\,\cap\,V_{\ge j+1}\bigr|
  \;\le\;
  d.
\]

Recall $V_{\ge j+1} = V(T)\setminus \bigl(V_1 \cup \cdots \cup V_j\bigr)$, so if $c \in \mathrm{children}_T(x) \cap V_{\ge j+1}$, then $\ell_T(c)\ge j+1$.  
By validity of the mapping, $\mathrm{map}(c)\in N_G(\mathrm{map}(x))$.  
Moreover, if $\ell_T(c)\ge j+1$ but $\ell_G(\mathrm{map}(c))$ were $< j+1$, then $c$ would be strictly monotonic reachable and the induction hypothesis would force $\ell_T(c)\le j$, a contradiction.  
Thus we must have $\ell_G(\mathrm{map}(c)) \ge j+1$.  
Hence each $c \in \mathrm{children}_T(x)\cap V_{\ge j+1}$ corresponds to a neighbor of $\mathrm{map}(x)$ in $G$ whose $\ell_G$-value is at least $j+1$.  

Because $\ell_G(\mathrm{map}(x))=j+1$ and for all $v\in V(G)$ with $\ell_G(v)\neq \infty$, 
\[
  \bigl|\{\,u\in N_G(\mathrm{map}(x)) : \ell_G(u)\ge \ell_G(\mathrm{map}(x))\}\bigr|
  \;\le\;
  d,
\]
it follows immediately that
\[
  \bigl|\mathrm{children}_T(x) \cap V_{\ge j+1}\bigr|
  \;\le\;
  d.
\]
Thus
\[
  \bigl|\mathrm{children}_T(x)\cap V_{\ge j+1}\bigr|
  \;+\;
  \bigl|\texttt{Missing}(x)\bigr|
  \;\le\;
  d + \textit{missing}
  \;\le\;
  a.
\]
Therefore if $x \in V_{\geq j+1}$, then $x$ is placed in $V_{j+1}$ by the algorithm, so $\ell_T(x)=j+1$. This completes the inductive step.
\end{proof}
}
\begin{lemma}
\label{lem:partiallayerassignmenttree}
Let $G$ be a graph, $d, L \in \mathbb{N}$ and $\ell_G \colon V(G) \mapsto [L] \cup \{\infty\}$ such that for every $v \in V(G)$ with $\ell_G(v) \neq \infty$, $|\{u \in N_G(v) \colon \ell_G(u) \geq \ell_G(v)\}| \leq d$.

Furthermore, let $B,k,s \in \mathbb{N}_0$ be parameters that satisfy $k \geq d$ and $s > \log_2(L)$.

Let $v \in V(G)$ with $\ell_G(v) \neq \infty$ and
$NumPathsIn_{G,\ell_G}(v) \leq \sqrt{B}.$

Let $T_v^{(s)}$ and $map_v^{(s)}$ be the rooted tree and mapping that one obtains by calling

$ExponentiateAndLocalPrune(G, B, k, s)$ and let $\ell_v \colon V(T_v^{(s)}) \mapsto  [L] \cup \{\infty\}$ be the layer assignment one obtains from calling $PartialLayerAssignmentTree(G, T_v^{(s)}, map_v^{(s)}, a:=k \cdot(s+1), L)$. Then, if $r_v$ denotes the root of $T_v^{(s)}$, it holds that $\ell_v(r_v) \leq \ell_G(v)$.
\end{lemma}
\begin{proof}
\cref{lem:missing_upper_bound} gives that every $x \in V(T_v^{(s)})$ that is strictly monotonically reachable with respect to $\ell_G$ satisfies 

$|Missing_{G,T_v^{(s)}, map_v^{(s)}}(x)| \leq s \cdot k$. Thus, we can use \cref{lem:layer_assignment_tree} with $missing := s \cdot k$ and $a := (s+1) \cdot k$ to conclude that for every $x \in V(T_v^{(s)})$ that is strictly monotonically reachable with respect to $\ell_G$, we have $\ell_v(x) \leq \ell_G(map(x))$. As the root $r_v$ of $T_v^{(s)}$ is strictly monotonically reachable with respect to $\ell_G$, this implies that $\ell_v(r_v) \leq \ell_G(v)$.
\end{proof}

\begin{lemma}
\label{lem:upper_bound_out_degree}
Let $G$ be a graph, $T$ be a rooted tree, $map \colon V(T) \mapsto V(G)$ be a \emph{valid} mapping, and $a, L \in \mathbb{N}$. Let $\ell_T \colon V(T) \mapsto [L] \cup \{\infty\}$ be the layer assignment computed by

$PartialLayerAssignmentTree(G, T, map, a, L)$ and $\ell \colon V(G) \mapsto [L] \cup \{\infty\}$ be defined as $\ell(v) = \min \{\ell_T(x) \colon x \in V(T) \text{ and } map(x) = v\}$, with the cornercase definition of $\ell(v) = \infty$ if no node maps to $v$. Then, for every $v \in V(G)$ with $\ell(v) \neq \infty$, it holds that 

\[|\{u \in N_G(v) \colon \ell(u) \geq \ell(v)\}| \leq a.\]
\end{lemma}
\shortOnly{
The proof of \cref{lem:upper_bound_out_degree} can be found in the full version of the paper.
}
\fullOnly{
\begin{proof}
Let $v \in V(G)$ and $j \in \mathbb{N}$ be arbitrary with $\ell(v) = j$. By definition of $\ell(v)$, there exists some node $x \in V(T)$ with $map(x) = v$ and $\ell_T(x) = j$. In particular, it holds that $x \in V_j$ and therefore

\begin{align*}
|children_T(x) \cap V_{\geq j}| + |Missing(x)| \leq a.
\end{align*}

Let $u \in N_G(v)$ with $\ell(u) \geq j$. First, consider the case that there exists some child $c \in children_T(x)$ with $map(c) = u$. Then, $\ell_T(c) \geq \ell(u) \geq j$ and therefore $c \in children_T(x) \cap V_{\geq j}$.
Otherwise, if there is no child $c \in children_T(x)$ with $map(c) = u$, then $u \in Missing(x)$. Therefore,

\begin{align*}
|\{u \in N_G(v) \colon \ell(u) \geq j\}| \leq |children_T(x) \cap V_{\geq j}| + |Missing(x)| \leq a.
\end{align*}
\end{proof}
}

\fullOnly{\begin{algorithm}[H]}
\shortOnly{\begin{algorithm}}
\SetAlgoLined
\KwIn{Graph $G$, non-negative integers $ B, k, L, s \in \mathbb{N}$}
\KwOut{A layer assignment $\ell : V(G) \mapsto [L] \cup \{\infty\}$}
$(T_v^{(s)}, map_v^{(s)})_{v \in V(G)} \gets ExponentiateAndLocalPrune(G, B, k, s)$ \\
\For{every $v \in V(G)$}{
$\ell_v \colon V(T_v^{(s)}) \mapsto [L] \cup \{\infty\} \gets PartialLayerAssignmentTree(G, T_v^{(s)}, map_v^{(s)}, (s+1) \cdot k, L)$ 

}
$\forall u \in V(G)$, $\ell(u) \gets \min_{v \in V(G), x \in V(T^{(s)}_v), map_v^{(s)}(x) = u}\ell_v(x)$ 

\Return $\ell$
\caption{PartialLayerAssignment}
\label{alg:partial_layer_assignment}
\end{algorithm}
\begin{claim}
\cref{alg:partial_layer_assignment} can be implemented in $O(s)$ MPC rounds with $O(n^\delta + B)$ words of local memory and $O(nB+m)$ words of global memory, where $n$ is an upper bound on the number of vertices and $m$ is an upper bound on the number of edges of the graph $G$.
\label{claim:MPC_implementation_PLA}
\end{claim}
\begin{proof}
This claim directly follows from \cref{claim:MPC_implementation} together with a straightforward implementation using standard MPC primitives developed in previous works, e.g. [\cite{andoni2018parallelgraphconnectivitylog}, Section E] and the references therein.
\end{proof}

\begin{claim}
\label{claim:out_degree}
Let $G$ be a graph, $B, k, L, s \in \mathbb{N}$ and $\ell \colon V(G) \mapsto [L] \cup \{\infty\}$ be the partial layer assignment computed by 

$PartialLayerAssignment(G, B, k, L,s)$. Then, for every $v \in V(G)$ with $\ell(v) \neq \infty$, it holds that
\[|\{u \in N_G(v) \colon \ell(u) \geq \ell(v)\}| \leq (s+1) \cdot k.\]
\end{claim}
\shortOnly{
A formal proof of this claim can be found in the full version of the paper.
}
\fullOnly{
\begin{proof}
Let \(v \in V(G)\) be arbitrary with \(\ell(v) \neq \infty\). By the definition of \(\ell\), there exists some vertex \(v' \in V(G)\) and a node 
$
x \in V(T_{v'}^{(s)})$ such that 
$
map_{v'}^{(s)}(x)=v$ and $\ell_{v'}(x)=\ell(v)$. Let
\[
\widetilde{\ell}_{v'}\colon V(G) \to [L] \cup \{\infty\}
\]
be defined by setting for every \(u\in V(G)\),
\[
\widetilde{\ell}_{v'}(u) \;=\; 
\begin{cases}
\min\{\,\ell_{v'}(y) \colon y \in V(T_{v'}^{(s)}) \text{ with } map_{v'}^{(s)}(y)=u\}, & \text{if such } y \text{ exists},\\[1mm]
\infty, & \text{otherwise}.
\end{cases}
\]
By the definition of the overall assignment \(\ell\) we have 
\[
\ell(u)=\min\{\,\ell_{w}(y) \colon w \in V(G),\; y \in V(T_{w}^{(s)}),\; map_{w}^{(s)}(y)=u\} \;\le\; \widetilde{\ell}_{v'}(u)
\]
for every \(u\in V(G)\). Hence, for any \(u\in N_G(v)\) satisfying
$\ell(u)\ge \ell(v)=\ell_{v'}(x)$,
we obtain
$
\widetilde{\ell}_{v'}(u) \,\ge\, \ell_{v'}(x).
$
Thus,
\[
\{\,u\in N_G(v):\,\ell(u)\ge \ell(v)\,\}\;\subseteq\;\{\,u\in N_G(v):\,\widetilde{\ell}_{v'}(u)\ge \ell_{v'}(x)\,\}.
\]

Now, by the assumptions in the call to \(\textsc{PartialLayerAssignmentTree}\) and by Lemma~\ref{lem:upper_bound_out_degree} (applied to the tree \(T_{v'}^{(s)}\) and its valid mapping \(map_{v'}^{(s)}\)), we know that
\[
\Bigl|\{\,u\in N_G(v):\,\widetilde{\ell}_{v'}(u)\ge \ell_{v'}(x)\,\}\Bigr| \;\le\; (s+1)\cdot k.
\]
It follows that
\[
\Bigl|\{\,u\in N_G(v):\,\ell(u)\ge \ell(v)\,\}\Bigr| \;\le\; (s+1)\cdot k.
\]
This completes the proof of the claim.
\end{proof}
}
\begin{lemma}
\label{lem:partial_assignment}
Let \(G\) be a graph and let \(B,k\in\mathbb{N}\) satisfy
$
k^{100} \le B \le n^{\delta/100} \quad \text{and} \quad k \ge 100\cdot \lambda(G)$,
where \(n\) and \(m\) are upper bounds on the number of vertices and edges of \(G\), respectively. Then there exists an MPC algorithm that runs in \(O(\log\log n)\) rounds with \(O(n^\delta)\) local memory and \(O(nB + m)\) words of global memory, and which computes a partial layer assignment
$
\ell\colon V(G) \to \{1,2,\dots,\lceil 0.1\log_k(B)\rceil\}\cup\{\infty\}
$
with the following properties:
\begin{enumerate}
    \item The out-degree of \(\ell\) is at most \(\lceil 100\log\log(n)\rceil\cdot k\).
    \item For every \(j\in \{1,2,\dots,\lceil 0.1\log_k(B)\rceil\}\),
    $
    \Bigl|\{v\in V(G) : \ell(v)\ge j\}\Bigr| \le 0.5^{\,j-1}\cdot |V(G)|.
    $
\end{enumerate}
\end{lemma}
\begin{proof}
We begin by executing the algorithm
\[
\textsc{PartialLayerAssignment}(B,k,L,s)
\]
with parameters
$
L := \lceil 0.1\log_k(B)\rceil \quad \text{and} \quad s :=\lceil 10\log\log(n)\rceil.
$
Let \(\ell\) denote the resulting partial layer assignment. Then, by \cref{claim:out_degree}, the out-degree of \(\ell\) is at most
\[
(s+1)\cdot k = \bigl(\lceil 10\log\log(n)\rceil + 1\bigr) k \le \lceil 100\log\log(n)\rceil\, k,
\]
which establishes property (1).

It remains to prove that for every \(j\in \{1,2,\dots, L\}\) we have
\[
\Bigl|\{v\in V(G):\ell(v)\ge j\}\Bigr| \le 0.5^{\,j-1}\cdot |V(G)|.
\]
To this end, define the auxiliary layer assignment \(\ell_G\colon V(G)\to \{1,2,\dots, L\}\cup\{\infty\}\) as follows. In each of the \(L\) iterations, remove from the current graph all vertices of degree at most \(k\) and assign to each removed vertex the index of the iteration in which it was removed. A straightforward calculation shows that, using the assumption \(k\ge 100\cdot\lambda(G)\), for every \(j\in \{1,\dots,L\}\) we have
\[
\Bigl|\{v\in V(G):\ell_G(v)\ge j\}\Bigr| \le 0.1^{\,j-1}\cdot |V(G)|.
\]

Now, fix any \(j\in \{2,3,\dots,L\}\). By \cref{lem:partiallayerassignmenttree} together with \cref{lem:double_counting} we obtain
\begin{align*}
\Bigl|\{v\in V(G):\ell(v)\ge j\}\Bigr| 
&\le \Bigl|\{v\in V(G):\ell_G(v)\ge j\}\Bigr| \shortOnly{\\
&+} \Bigl|\Bigl\{v\in V(G): \operatorname{NumPathsIn}_{G,\ell_G}(v)\ge \sqrt{B}\Bigr\}\Bigr|\\[1mm]
&\le 0.1^{\,j-1}\cdot |V(G)| + \frac{k^L}{\sqrt{B}}\cdot |V(G)|\\[1mm]
&\le 0.5^{\,j-1}\cdot |V(G)|.
\end{align*}
\end{proof}
\begin{lemma}
\label{lem:partial_assignment_2}
Let \(G\) be a graph and let \(B,k\in\mathbb{N}\) satisfy
$
k^{100} \le B \le n^{\delta/100} \quad \text{and} \quad k \ge 100\cdot \lambda(G)$,
where \(n\) and \(m\) are upper bounds on the number of vertices and edges of \(G\), respectively. Then there exists an MPC algorithm that runs in \(O(\log(k)\log\log (n))\) rounds with \(O(n^\delta)\) local memory and \(O(nB + m)\) words of global memory, and which computes a partial layer assignment
$
\ell\colon V(G) \to \{1,2,\dots,\lceil 100\log(B)\rceil\}\cup\{\infty\}
$
with the following properties:
\begin{enumerate}
    \item The out-degree of \(\ell\) is at most \(\lceil 100\log\log(n)\rceil\cdot k\).
    \item For every \(j\in  \{1,2,\dots,\lceil 100\log(B)\rceil\),
    $
    \Bigl|\{v\in V(G) : \ell(v)\ge j\}\Bigr| \le 0.5^{\,j-1}\cdot |V(G)|.
    $
\end{enumerate}
\end{lemma}
\begin{proof}
The main idea is to iteratively apply the algorithm of \Cref{lem:partial_assignment} on the unassigned vertices. We now describe the process in detail.

\medskip
\noindent\textbf{Setup.}  
Let
\[
T \;=\; \Bigl\lceil \frac{\lceil100\log(B)\rceil}{\lceil 0.1\log_k(B)\rceil} \Bigr\rceil \;=\; O(\log k)
\]
be the number of iterations we will perform. We work in rounds indexed by \(i=1,\dots, T\). In round \(i\) we consider the subgraph induced by the vertices that are still unassigned (i.e. currently assigned \(\infty\)). Denote these vertices by
\[
U_i \;=\; \{ v\in V(G): \ell(v)=\infty \}\,,
\]
with the convention that \(U_1=V(G)\).

\medskip
\noindent\textbf{Iteration.}  
In round \(i\), we run the MPC algorithm from \Cref{lem:partial_assignment} on the subgraph \(G[U_i]\) with parameters \(B\) and \(k\). This produces a partial assignment
\[
\ell_i \colon U_i \to \{1,2,\dots,\lceil 0.1\log_k(B)\rceil\}\cup\{\infty\}
\]
satisfying:
\begin{enumerate}
    \item The out-degree of \(\ell_i\) is at most \(\lceil 100\log\log(n)\rceil\cdot k\).
    \item For every \(j\in \{1,2,\dots,\lceil 0.1\log_k(B)\rceil\}\),
    $
    \Bigl|\{v\in U_i: \ell_i(v)\ge j\}\Bigr|\le 0.5^{\,j-1}|U_i|.
    $
\end{enumerate}
We then assign a final layer to each vertex \(v\in U_i\) that receives a finite layer in round \(i\) by setting
\[
\ell(v) \;=\; (i-1)\cdot \lceil 0.1\log_k(B)\rceil \;+\; \ell_i(v).
\]
(assuming $(i-1)\cdot \lceil 0.1\log_k(B)\rceil \;+\; \ell_i(v) \leq \lceil 100 \log(B)\rceil)$.
Vertices for which \(\ell_i(v)=\infty\) remain unassigned and are carried over to the next round (i.e. they belong to \(U_{i+1}\)).

\medskip
\noindent\textbf{Out-Degree.}  
Each round produces an assignment \(\ell_i\) with an out-degree at most \(\lceil 100\log\log(n)\rceil\cdot k\), and therefore the out-degree of the final assignment also is at most \(\lceil 100\log\log(n)\rceil\cdot k\).

\medskip
\noindent\textbf{Decay Property.}  
Fix an integer \(j\) with \(1\le j\le \lceil 100\log(B)\rceil\). Write
\[
j \;=\; (i-1)\cdot \lceil 0.1\log_k(B)\rceil + j'\,,
\]
where \(1\le j'\le \lceil 0.1\log_k(B)\rceil\) and \(i\in \{1,2,\dots,T\}\). In round \(i\), the decay property for \(\ell_i\) implies that at most a fraction \(0.5^{\,j'-1}\) of the vertices in \(U_i\) receive a value at least \(j'\). Furthermore, the process that carries over unassigned vertices from round \(i-1\) to round \(i\) ensures that
\[
|U_i| \;\le\; 0.5^{(i-1)\cdot \lceil 0.1\log_k(B)\rceil}\,|V(G)|.
\]
Thus, the number of vertices \(v\) with final layer
\[
\ell(v)\ge j \text{ (i.e., those with either \(v\in U_i\) or with \(\ell_i(v)\ge j'\) in round \(i\))},
\]
is at most
\[
0.5^{\,j'-1}\cdot 0.5^{(i-1)\cdot \lceil 0.1\log_k(B)\rceil}\cdot |V(G)| \;\le\; 0.5^{\,j-1}|V(G)|.
\]

\medskip
\noindent\textbf{Round Complexity and Memory.}  
Since each invocation of the algorithm from \Cref{lem:partial_assignment} runs in \(O(\log\log n)\) rounds and we perform \(T=O(\log k)\) sequential rounds, the overall round complexity is
\[
O\bigl(\log(k)\log\log n\bigr).
\]
The local memory and global memory bounds remain \(O(n^\delta)\) and \(\tilde{O}(nB+m)\) respectively, as in \Cref{lem:partial_assignment}.

\medskip
Thus, the final partial layer assignment
\[
\ell\colon V(G) \to \{1,2,\dots,\lceil 100\log(B)\rceil\}\cup\{\infty\}
\]
satisfies the required out-degree and decay properties, and the MPC algorithm runs in \(O(\log(k)\log\log n)\) rounds with the stated memory bounds.
\end{proof}
\begin{lemma}
\label{lem:complete_layer_assignment}
Let \(G\) be a graph and let \(k\in\mathbb{N}\) satisfy
$ k \ge 100\cdot \lambda(G)$,
where \(n\) and \(m\) are upper bounds on the number of vertices and edges of \(G\), respectively. Then there exists an MPC algorithm that runs in \(O(\log(k)\log^2\log (n))\) rounds with \(O(n^\delta)\) local memory and \(O(n + m)\) words of global memory, and which computes a complete layer assignment
$
\ell\colon V(G) \to \mathbb{N}
$
with the following properties:
\begin{enumerate}
    \item The out-degree of \(\ell\) is at most \(\lceil 100\log\log(n)\rceil\cdot k\).
    \item For every \(j\in  \mathbb{N}\),
    $
    \Bigl|\{v\in V(G) : \ell(v)\ge j\}\Bigr| \le 0.5^{\,j-1}\cdot |V(G)|.
    $
\end{enumerate}
\end{lemma}
\begin{proof}
The complete layer assignment is obtained by first “peeling” the graph to reduce the number of vertices, and then iteratively applying the partial assignment algorithm of \Cref{lem:partial_assignment_2} with a boosting procedure. We describe the process in three stages.

\medskip
\textbf{Stage 1. Initial Peeling.}  
In order to eventually run the partial assignment algorithm with a large budget per vertex, we first perform $\lceil 100 \log(k)\rceil$ rounds of a simple peeling procedure on \(G\) as follows. In each peeling round, remove from the current graph all vertices of degree at most \(k\). Standard arguments show that in every such round at least half of the vertices are removed. Hence, after $\lceil 100 \log(k)\rceil$ rounds, the set \(U\) of remaining vertices satisfies
\[
|U| \;\le\; \frac{n}{k^{100}}.
\]
Thus, on the subgraph \(G[U]\) the effective per-vertex budget is at least \(B_0 = k^{100}\).

\medskip
\textbf{Stage 2. Partial Assignment with Budget Boosting.}  
We now work on the subgraph \(G[U]\) (with \(U\subseteq V(G)\)) and iteratively assign layers to its vertices. In each \emph{phase} \(i\) we maintain a current set \(U_i\) of vertices that have not yet been assigned a finite layer, and a budget \(B_i\) (initially, \(U_0=U\) and \(B_0=k^{100}\)). In phase \(i\), we run the MPC algorithm from \Cref{lem:partial_assignment_2} on \(G[U_i]\) with parameters \(B_i\) and \(k\). This produces a partial layer assignment 
\[
\ell_i\colon U_i \to \{1,2,\dots,\lceil 100\log(B_i)\rceil\}\cup\{\infty\},
\]
which satisfies:
\begin{enumerate}
    \item The out-degree is at most \(\lceil 100\log\log(n)\rceil\cdot k\).
    \item For every \(j\in\{1,\dots,\lceil 100\log(B_i)\rceil\}\),
    \[
    \Bigl|\{v\in U_i:\ell_i(v)\ge j\}\Bigr|\le 0.5^{\,j-1}|U_i|.
    \]
\end{enumerate}
In particular, setting \(j=\lceil 100\log(B_i)\rceil\) we deduce that the number of vertices that remain unassigned after phase \(i\) is
\[
|U_{i+1}| \;\le\; \frac{|U_i|}{B_i^{100}}.
\]
For every vertex \(v\in U_i\) that receives a finite layer (i.e. \(\ell_i(v)\neq \infty\)), we define its final layer as
\[
\ell(v) \;=\; L_i + \ell_i(v),
\]
where \(L_i\) is an offset chosen to ensure that layers assigned in different phases are disjoint. For example, one may set
\[
L_i \;=\; \lceil 100\log(k)\rceil + \sum_{t=0}^{i-1}\lceil 100\log(B_t)\rceil.
\]

After phase \(i\), we \emph{boost} the budget by setting
\[
B_{i+1} \;=\; \min(B_i^{100}, n^\delta).
\]

After $O(\log \log n)$ phases, we have $U_T = \emptyset$, so every vertex in $U$ (and hence in $V(G)$) has been assigned a finite layer.

\medskip
\textbf{Stage 3. Complexity and Verification.}  
In Stage 1, the peeling procedure takes \(O(\log k)\) rounds. In Stage 2, each phase runs the partial assignment algorithm from \Cref{lem:partial_assignment_2} in \(O(\log(k)\log\log(n))\) rounds. Since there are \(T = O(\log \log n)\) phases, the total round complexity is
\[
O\Bigl(\log k + \log \log(n) \cdot \log(k)\log\log(n)\Bigr) = O\bigl(\log(k)\log^2\log(n)\bigr).
\]
Furthermore, in each phase the local memory requirement is \(O(n^\delta)\) and the global memory used is \(O(|U_i|B_i + m) = O(n+m)\) words. 
Finally, the final complete layer assignment \(\ell\colon V(G)\to \mathbb{N}\) inherits the following properties from the partial assignments:
\begin{enumerate}
    \item The out-degree is at most \(\lceil 100\log\log(n)\rceil\cdot k\).
    \item For every \(j\in\mathbb{N}\), a simple induction over the phases shows that
    \[
    \Bigl|\{v\in V(G) : \ell(v)\ge j\}\Bigr|\le 0.5^{\,j-1}\cdot |V(G)|.
    \]
\end{enumerate}

This completes the proof.
\end{proof}

\MainThm*
\begin{proof}
Using an extra $O(\log(n))$ factor in the global memory, we can assume that we are given $k$ with $k \in [100\lambda(G), 200 \lambda(G)]$. First, consider the case that $k = \log(n)^{O(\log \log n)}$. In that case, we can use \cref{lem:partial_assignment_2} to compute a complete layering $\ell$ with out-degree $O(k \log \log n)$ in $\poly(\log \log n)$ rounds. Then, we can get the desired orientation with out-degree $O(\lambda(G) \log \log n)$ by orienting edges from the lower to the higher layer, breaking ties arbitrarily. If $k \geq \log(n)^{O(\log \log n)}$, we start by computing a random edge partitioning into $P = \Theta(k/log(n))$ parts $E_1, E_2, \ldots, E_P$. \cref{lem:edge_partition} gives that the arboricity of the graph $G_i = (V(G),E_i)$ is $O(\log n)$. Hence, we can use \cref{lem:partial_assignment_2} to orient the edges in $E_i$ with outdegree $O(\log n \log \log n)$ in $\poly(\log n)$ rounds. Then, combining the $P$ edge orientations gives an edge orientation of $G$ with out-degree $O(\lambda(G) \log \log n)$.
\end{proof}

\fullOnly{

\section{Coloring Algorithm}

\ColoringThm*
\begin{proof}
We present the proof in several steps. Using an extra $O(\log(n))$ factor in the global memory, we can assume that we are given $k$ with $k \in [100\lambda(G), 200 \lambda(G)]$. 

\bigskip
\noindent \textbf{Random partitioning, if needed.} First, we transform our problem into coloring in graphs with arboricity at most $O(\log n)$: If $k\leq O(\log n)$, there is nothing to do here. Otherwise, we apply a random vertex partitioning, splitting vertices into $P=\Theta(k/\log n)$ parts $V_1, V_2, \ldots, V_P$. By \Cref{lem:vertex_partition}, we know that with high probability, the subgraph induced by each part has arboricity $O(\log n)$. We then will solve the coloring problem of each part separately, using different palettes of size $O((k/P) \log\log n)$, hence obtaining an $O(\lambda \log\log n)$ coloring for the overall graph. Given this, for the rest of this proof, we focus on coloring a graph where the arboricity is at most $O(\log n \log \log n)$. 

\bigskip
\noindent \textbf{Layering and orientation.} We apply \Cref{lem:complete_layer_assignment} to compute the $L=\Theta(\log n)$-layer partition $V=H_1\sqcup H_2 \ldots \sqcup H_{L}$ with outdegree at most $d=O(\lambda \log\log n)$, satisfying the property that $\cup_{i=j}^L|H_i| \leq 2^{-j+1} \cdot n$ for every $j \in [L]$. This works in at most $O(\log^3\log n)$ rounds of MPC. We color this graph using $3d$ colors. 

\bigskip
\noindent \textbf{LOCAL model coloring.} Given the layering computed, from the viewpoint of the LOCAL model, we could color this graph in $\Theta(\log n)$ phases, where in phase $i$, the task is to compute a coloring for vertices in $H_{L-i+1}$, which for each node $v$ avoids the colors of its neighbors in higher layers. This is simply a  $degree+1$ list coloring problem, on the graph induced by $H_i$ (in fact, the number of available colors will be at least $2d$ where $d=\Theta(\lambda \log\log n)$). This scheme would start by coloring the last layer $H_L$, it would proceed layer by layer, and would end with coloring the nodes in $H_1$. Each phase would take $O(\log^{1.67}\log n)$ LOCAL rounds, thanks to the state of the art LOCAL model algorithm~\cite{halldorsson2022near,ghaffari2024near}. Hence, the overall algorithm would take $O(\log n \cdot \log^{1.67}\log n)$ rounds.

\bigskip
\noindent \textbf{Faster simulation in MPC.} We simulate the above LOCAL algorithm much faster in MPC, using a \textit{directed graph exponentiation}: Suppose edges inside each layer are duplicated as bidirectional edges, and edges across layers are directed toward the higher layer endpoint. 
Also, assume we have already determined the color of each node in $H_j \sqcup H_{j+1} \sqcup \ldots \sqcup H_{L}$. We would now like to determine the color of each node in $H_{j'} \sqcup H_{j'+1} \sqcup \ldots \sqcup H_{j-1}$, for some $j'$. To do so, it suffices for each node $v \in H_{j'} \sqcup H_{j'+1} \sqcup \ldots \sqcup H_{j-1}$  to learn all nodes that are reachable from $v$ along directed paths of distance at most $O\left((j-j')\log^{1.67}\log(n)\right)$, along with their color in case they belong to layers $H_j \sqcup H_{j+1} \sqcup \ldots \sqcup H_{L}$. Using that $d = O(\lambda \log \log n)$ and $\lambda = O(\log n)$, each node $v \in  H_{j'} \sqcup H_{j'+1} \sqcup \ldots \sqcup H_{j-1}$ needs to learn at most $d^{O\left((j-j')\log^{1.67}\log(n)\right)} = 2^{O\left((j-j')\log^{2.67}\log(n)\right)}$ many nodes. In particular, as long as $j-j' = \Theta(\delta\log(n)/\log^{2.67}\log(n))$, with a suitably small constant, each node $v$ needs to learn $O(n^\delta)$ nodes and thus the local memory constraint is satisfied. In order to satisfy the global memory constraint, it suffices to ensure that

\[\frac{n}{2^{j'}} \cdot 2^{O\left((j-j')\log^{2.67}\log(n)\right)} = O(n),\]

as $\cup_{i=j'}^L|H_i| \leq 2^{-j'+1} \cdot n$. This constraint is satisfied as long as

\[j-j' = O\left(\frac{\delta\cdot j}{\log^{2.67} \log (n)}\right).\]

Thus, by using directed graph exponentiation along outgoing edges (see \cite[Definition 3.3]{latypov2021coloring} for lower level details, or \Cref{lem:directedExp} for a sketch), we can compute in $O(\log \log n)$ rounds the color of a $\frac{1}{\Theta(\log^{2.67} \log n)}$-fraction of layers for which we have not yet computed their color, assuming that $j \geq \Theta(\log^{2.67} \log n)$. Since there are $O(\log n)$ layers in total, we need to repeat this $O(\log^{3.67} \log n)$ times, which takes $O(\log^{4.67} \log n)$ MPC rounds. Once $j\leq \Theta(\log^{2.67} \log n)$, we simply run the LOCAL model algorithm phase by phase, without a speed up. and that takes $\Theta(\log^{2.67} \log n) \cdot \Theta(\log^{1.67} \log n)= \Theta(\log^{4.34} \log n)$ additional rounds. Hence, the complete round complexity is $O(\log^{4.67} \log n)$.

\end{proof}

\begin{lemma}[Directed exponentiation]\label{lem:directedExp} For any constant $\delta>0$, we have the following: Suppose each node $v$ has an information bundle $B_v$ of $b_v\leq n^{\delta/2}$ words, which others might be interested to learn, and each node $u$ is interested in receiving the information bundles from a list $L_u$ of nodes. Suppose also that (A) $|L_u|\leq n^{\delta/2}$ and (B) $\sum_{u\in V} \sum_{v\in L_u} |b_v| = O(m+n).$ Then, there is an MPC algorithm with $n^{\delta}$ local memory and $O(m+n)$ global memory that performs this task in $O(1)$ rounds.
\end{lemma}
\begin{proof}[Proof sketch]
First, we make each node $v$ learn how many nodes $u$ have $v\in L_u$. For that, we apply a constant-round sorting algorithm: $v$ initiates two items $(v.ID, -\infty)$ and $(v.ID, +\infty)$ and each node $u$ that has $v\in L_u$ initiates an item $(v.ID,u)$. Then, we sort all these items lexicographically, using any standard constant round sorting algorithm~\cite{karloff2010mpc, im2023massively, GhaffariMPCNotes}. The difference in the ranks of the two items $(v.ID, -\infty)$ and $(v.ID, +\infty)$, which will be known to the machine that holds $v$, tells $v$ how many nodes $u$ have $v\in L_u$. Let this number be $k_v$.

Second, for each $v$, we generate $k_v$ copies of the information bundle $B_v$, identified with numbers $1$ to $k_v$. We do this simultaneously for all $v$, in constant rounds, using a standard broadcast tree for each of them that increases the number of copies of $B_v$ from $1$ to $k_v$, per iteration by an $n^{\delta/2}$ factor. See e.g., \cite[Section 1.3.2]{GhaffariMPCNotes}. 

Third, we sort two things: (1) an item list $(u, v)$ for each node $u$ and each node $v\in L_u$, indicating that $u$ wants to receive a copy of $B_v$. (2) an item list $(v, i)$ for the $i^{th}$ copy of $B_v$ generated above. Then, we simply perform a matching where node $u$ that has rank $j$ in the first list contacts the machine that holds rank $j$ in the second list and receives the appropriate copy of $B_v$ from it.
\end{proof}

}